\documentclass[lettersize,journal]{IEEEtran}
\usepackage{amsmath,amsfonts}
\usepackage{algorithmic}
\usepackage{array}
\usepackage{textcomp}
\usepackage{stfloats}
\usepackage{url}
\usepackage{verbatim}
\usepackage{graphicx}
\usepackage{caption}
\usepackage{subcaption}
\usepackage{float}
\usepackage{color,soul}

\usepackage{amsmath}
\usepackage{amsthm}
\usepackage{algorithm}

\usepackage[utf8]{inputenc} 
\usepackage[T1]{fontenc}    
\usepackage{url}            
\usepackage{booktabs}       
\usepackage{amsfonts}       
\usepackage{nicefrac}       
\usepackage{tcolorbox}
\usepackage{microtype}      
\usepackage{amsmath}
\usepackage{amsthm}
\usepackage{graphicx}
\usepackage{subcaption}
\usepackage{algorithm}
\usepackage{algorithmic}
\usepackage{wrapfig}

\usepackage{xr}

\newtheorem{Theorem}{Theorem}
\newtheorem{Proposition}{Proposition}
\newtheorem{Lemma}{Lemma}
\newtheorem{Corollary}{Corollary}
\newtheorem{Definition}{Definition}

\usepackage{amsmath,amsfonts}
\usepackage{algorithmic}
\usepackage{array}
\usepackage{textcomp}
\usepackage{stfloats}
\usepackage{url}
\usepackage{verbatim}
\usepackage{graphicx}
\def\BibTeX{{\rm B\kern-.05em{\sc i\kern-.025em b}\kern-.08em
    T\kern-.1667em\lower.7ex\hbox{E}\kern-.125emX}}
\usepackage{balance}
\begin{document}
\title{A Theory of I/O-Efficient Sparse Neural Network Inference}
\author{Niels Gleinig, Tal Ben-Nun, Torsten Hoefler\\
\{niels.gleinig, talbn, torsten.hoefler\}@inf.ethz.ch\\
ETH Zürich
}

\markboth{}
{A Theory of I/O-Efficient Sparse Neural Network Inference}

\maketitle

\begin{abstract}
As the accuracy of machine learning models increases at a fast rate, so does their demand for energy and compute resources. On a low level, the major part of these resources is consumed by data movement between different memory units. Modern hardware architectures contain a form of \textit{fast memory} (e.g., cache, registers), which is small, and a \textit{slow memory} (e.g., DRAM), which is larger but expensive to access. We can only process data that is stored in fast memory, which incurs data movement (input/output-operations, or \textit{I/Os}) between the two units. In this paper, we provide a rigorous theoretical analysis of the I/Os needed in sparse feedforward neural network (FFNN) inference. We establish bounds that determine the optimal number of I/Os up to a factor of $2$ and present a method that uses a number of I/Os within that range. Much of the \textit{I/O-complexity} is determined by a few high-level properties of the FFNN (number of inputs, outputs, neurons, and connections), but if we want to get closer to the exact lower bound, the instance-specific sparsity patterns need to be considered. Departing from the 2-optimal computation strategy, we show how to reduce the number of I/Os further with simulated annealing. Complementing this result, we provide an algorithm that constructively generates networks with maximum I/O-efficiency for inference. We test the algorithms and empirically  verify our theoretical and algorithmic contributions. In our experiments on real hardware we observe speedups of up to 45$\times$ relative to the standard way of performing inference.
\end{abstract}


\section{Introduction}

Almost all modern computing systems deploy different memory units: Some are fast but small and some are large but slow. Data can only be processed in fast memory. As the fast memory is very limited in size, data needs to be moved between fast and slow memory over the course of a computation. These data movements are called \textbf{I/Os} (abbreviating \textit{input/output-operations}). They are very expensive in time and energy. For example, on an NVIDIA V100 GPU, loading values from global memory costs \textit{up to 514.5$\times$ more clock cycles} than one fused-multiply-add (FMA) operation between two 32-bit floating point values (or 14$\times$ if stored on an L1 cache)~\cite{jia2018dissecting}.


%
In deep learning, and scientific computing in general, I/Os account for the major part of time and energy expenses~\cite{padal}. Hence, it is important to use I/Os efficiently by performing computations in a way that makes use of ``data locality'' and ``data reuse''. We need to schedule the computational steps in such a way that accesses to the same data are close together, sparing eviction of values from fast memory (e.g., cache) and reading them again.
Therefore, optimizing the implementation of DNN inference~\cite{szesurvey} and training~\cite{bennunsurvey} often requires tailor-made implementations for the available memory architectures in order to reach high utilization.

In this paper, we analyze and optimize I/Os with a theoretical model that applies across the wide variety of hardware in which this phenomenon occurs.
We consider exclusively sparse Feed-Forward Neural Networks (FFNNs) without shared connections. This model contains, but is not limited to, pruned Multi-Layer Perceptrons (MLP). We assume that the FFNN is given as a list of weighted edges in a Directed Acyclic Graph (DAG) together with one additional value for each vertex (being the input-value for input-neurons and the bias for non-input neurons). 
In this setting, the primary way to obtain I/O-efficiency is to ``use'' the connections in an efficient order. That is, when we use a connection, the value of the input neuron and the value of the partial sum of the output neuron should ideally already be in fast memory. Therefore, we need to find an order of the connections, in which the connections that have neurons in common, are close together.

Finding an order which is optimal with respect to this property is a combinatorial problem. We will show that if we use a traditional computational order, which corresponds to a natural interpretation of FFNN inference (for example, seeing it as a sequence of matrix vector multiplications) we can be far from optimal: The number of read-I/Os could be up to twice the optimal number, and the number of write-I/Os could differ from the optimal number by an arbitrarily large factor.

Computations in ML have some special features that we can use when aiming for I/O-efficiency. Unlike the research on I/O-efficient computing outside of ML, there are more degrees of freedom that we can exploit. For example, there are several acceptable solutions to a given problem: Two different FFNNs may produce different outputs, yet both can have a satisfying accuracy. It thus makes sense to adapt the neural architecture to the hardware. On the other hand, there is additional freedom on the side of the hardware, as it is not uncommon to search or even build hardware for specific ML applications~\cite{eie,tpu,szesurvey}.

Hence, our investigation is directed towards answering the following questions: (1) \emph{If we are given an FFNN and a fast memory of a given size, what is the \textbf{minimal number of I/Os} needed to perform inference on this FFNN?} (2) \emph{For a given FFNN, what is the \textbf{smallest memory size} with which we can perform inference with a minimal number of I/Os (i.e., we can avoid having to read and write intermediate results due to lack of memory)?} and (3)
\emph{For a fast memory of a given size, what are the \textbf{FFNN architectures} on which we can perform inference with a minimal number of I/Os?}

\subsection{Background and related work}

\paragraph{I/O-complexity}
There is an extensive line of theoretical research on I/O-efficient algorithms outside of machine learning.  
Hong \& Kung's seminal work~\cite{redblue} is among the first to investigate I/O-efficient algorithms formally. This work introduced a theoretical model called the \textit{red-blue pebble game} and used this model to show how standard matrix multiplication or FFT should be scheduled to use an asymptotically minimal number of I/Os. This model has been extended and altered in various ways. These extensions have been used to establish lower bounds and matching upper bounds for problems such as sorting, permutation, and matrix transposition~\cite{Aggar}; as well as sparse matrix dense vector multiplication~\cite{SPMDV}, and various graph algorithms~\cite{Muna}. However, these solutions are typically only optimal up to constant multiplicative factors and often these factors are large or not exactly known. In fact, it has been shown that finding exactly optimal solutions to the red-blue pebble game (corresponding to computation schedules with an exact minimal number of data movements) is a PSPACE-complete problem \cite{hardapprox}. Furthermore, this problem has been hypothesized to be even hard to approximate \cite{hardapprox}, since inapproximability is known for related pebble games~\cite{Liu1,Liu2}.

There are also bounds on the I/O-complexity of general computations with large proportions of inputs~\cite{redbluetrees}. These bounds apply to the problem considered in our paper. However, they are less fine-grained (they do not distinguish between read- and write-I/Os as we do) and generally less tight as the ones that we establish in this paper.

In ML-specific algorithms, I/O-efficiency has been practically identified as a crucial performance factor. Performance modeling of DNNs is usually derived from operation counts~\cite{paleo} and sums of fixed values obtained from benchmarking individual layers~\cite{chaos,oyama16}.
As for I/O complexity of DNN evaluation, the only theoretical work apart from this one, to the best of our knowledge, is given by Demmel \& Dinh~\cite{DemmelConv}, who provide communication bounds for convolutional and pooling layers.

\paragraph{Network pruning and sparsity}
DNN pruning (i.e., removing certain edges or nodes from a network while preserving accuracy) has been extensively studied over the last decades \cite{OBD,frankle,sparsesurvey}.
Pruning is of much practical interest for edge computing devices, as it allows to drastically reduce the memory footprint and number of floating point operations. Dense (fully-connected) layers have been a prime target for pruning, as they tend to account for a major proportion of the memory size. For example, in AlexNet~\cite{alexnet}, the final dense layers contain $90\%$ of the parameters.
Also transformers~\cite{attention}, which are currently very popular and achieve state-of-the-art results in NLP tasks, include FFNNs that comprise a large proportion of the overall size. For example, BERT~\cite{attention} includes several FFNNs of depth $2$ and weight matrices of dimensions $1024\times 4096$ and $4096\times 1024$ and the major part of BERTs parameters and compute time for inference come from these FFNNs~\cite{ivanov2020data}.

There are different ways to prune networks, e.g., by a given weight threshold~\cite{strom2015scalable,han2015deep}, the top-$k$ values~\cite{Aji_2017,sparcml}, or by way of gradual weight elimination~\cite{OBD}. Regardless of the method, pruning dense layers leads to sparse and unstructured networks, for which data-locality is more difficult to obtain.

\paragraph{Hardware Architectures}
There has been work on hardware-based optimizations for inference, especially for the class of sparse FFNNs, where Sze et al.~\cite{szesurvey} provide a detailed survey. Specific approaches include EDEN~\cite{EDEN}, which uses approximate DRAM for improving energy-efficiency; Bhardwaj et al.~\cite{infdist}, who consider communication in inference in a distributed setting; and EIE~\cite{eie} and Eyeriss~\cite{Sze} runs inference on a compressed CNN directly, using accelerated sparse matrix-vector multiplications and run-length encoding, respectively. While these approaches improve certain aspects of DNN processing, as the network size grows, such architectures ultimately also resort to multi-level memory hierarchies, which our approach addresses.

\subsection{Contributions}
Our main theoretical contributions are a collection of theorems and propositions that analyze the I/Os in FFNN inference. Theorem \ref{oinft} provides bounds on the number of I/Os that are needed to perform inference on a given FFNN-architecture. In Proposition \ref{tight}, we show that these bounds are optimal in the sense that none of them could be tightened by multiplying with any factor other than $1$. The lower and upper bounds on the total number of I/Os differ by a multiplicative factor of at most $2$. The proof of Theorem \ref{oinft} is constructive and shows how a computation can be done to use a number of I/Os within this 2-optimal range.
Departing from 2-optimality, our goal is to get closer to the \emph{exact} optimum by reordering the connections beneficially. We approach this problem with Simulated Annealing.
In Section \ref{genFNNs} we present a method for generating FFNNs, which according to Theorem \ref{compactGrowth} completely characterizes the architectures on which we can perform inference with a minimal number of I/Os for a given memory size (i.e., exactly matching the lower bound). Hence, this method can be used as a powerful tool to co-design neural networks and hardware. As a corollary of this theorem we obtain an upper bound on the smallest possible memory size that allows us to perform inference with maximal I/O-efficiency (that is, without having to write or read any temporary values). We test the algorithms on random FFNNs across a wide range of parameter settings (density, depth, width, memory sizes).

\section{Model and notations} \label{model}

We assume that initially all FFNN parameters and all input values are laid out in slow memory. These data are exactly the union of the following three types of data:
\begin{itemize}
\item The weighted connections.
\item One bias value for each non-input neuron.
\item One input value for each input neuron.
\end{itemize}
Each connection is described by an independent parameter triple $(i,j,w_{ij})$, where $i$ is the input neuron, $j$ the output neuron, and $w_{ij}$ the weight (there is no reuse or sharing of weights).
We denote the number of connections as $W$, the number of neurons $N$, the number of input neurons $I$, and the number of output neurons $S$.  
In our theoretical analysis, we assume the connections (the entire triples describing them), the numerical values at the neurons (biases, partial sums, outputs of the activation functions), and all other numerical values involved in the computation are data types of the same size.
Hence, the entire problem size is given by the total number of weights, biases, and input neurons: $W+(N-I)+I=W+ N$.
We have a fast memory that can hold at any given time a number of values (of this data type) that is given by the parameter $M$, where we assume $M\geq3$. In this model the single parameter $M$ defines the whole architecture: fast memory of size $M$ and slow memory of unlimited size. We can perform arbitrary computations ``for free'' on data that is stored in fast memory, including application of activation functions to those values. In particular, if a connection, the value of its input neuron, and the partial sum of its output neuron (which we consider initially the bias) are stored in fast memory, then we can add the product of the input value and the connection weight to the partial sum of the output neuron (and we assume that no additional memory is required for this operation). If we want to do a computation that requires a value that is currently only stored in slow memory, we first need to move this value from slow to fast memory. This movement counts as 1 \textbf{read-I/O}. Furthermore, if our memory is full, we first need to free up space before we can read new values. To do this, we can delete values. However, if we need a value again in future computation steps or if it is the value of an output neuron, we have to write the value to slow memory before deleting it. Deletions are for free whereas each write operation counts as 1 \textbf{write-I/O}. Initially, all $N+W$ parameters are stored in slow memory. The number of I/Os of a computation is the sum of read-I/Os and write-I/Os.
The goal is to perform inference computation, namely compute and store (i.e., write back to slow memory) the value of all output-neurons, with as few I/Os as possible.
We formally define the optimum that we aim for:

\begin{Definition}
For a given FFNN $\mathcal{N}$ and fast memory size $M$, we let \textbf{I/Os}$(\mathcal{N},M)\in \mathbb{N}$ denote the minimum number of overall I/Os that are needed to perform inference on $\mathcal{N}$ with a memory size $M$, where the minimum is taken over all possible computation strategies (i.e., sequences of computation-, read-, and write- steps, that solve this problem). Let \textbf{rI/Os}$(\mathcal{N},M)\in \mathbb{N}$ denote the minimum number of read-I/Os and \textbf{wI/Os}$(\mathcal{N},M)\in \mathbb{N}$ denote the minimum number of write-I/Os needed. \end{Definition}

\subsection{How can we describe inference-computations in this model?}

We introduce the concept of \emph{eviction policy}, and show that a computation corresponds to an \emph{eviction policy} together with a topological order of the connections.
An \textbf{eviction policy} is a set of rules or instructions that specify how we evict data from fast memory. That is, when the fast memory is full and we need to evict a value to create free space, the eviction policy determines which of the values to evict. \textbf{LRU} (least-recently-used) is the eviction policy defined by always evicting the value that has been used least recently among all values in fast memory. \textbf{RR} (round-robin) is the eviction policy defined by having a pointer specifying the value to be evicted next and moving this pointer one place to the right whenever we evict a value (and moving the pointer again to the first place of the memory when we reach the end). The eviction policy \textbf{MIN} (also known as Belady's optimal replacement algorithm) is defined by always evicting the value that will be referenced farthest in the future (if there are values that will not be used again, any of those is evicted). It has been shown that \textbf{MIN} uses the minimal number of I/Os for a given sequence of computation steps \cite{MINcache}. Notice that while it is difficult to implement \textbf{MIN} eviction policies for general computations, in the case of FFNN inference it is trivial to implement it offline once we fixed a topological order in which we process the weights (always evict the activations adjacent to weights farthest away in the given topological order).

We assume that an efficient eviction policy is provided. That is, when we evict a value that is either (1) already stored in slow memory (for example, when we read the value of a computed neuron to use it for its outgoing connections, but do not change it), or (2) a value that we will not need again in the future (a computed non-output-neuron has been used already for computing all neurons that depend on it), we simply delete it from fast memory, without spending a write-I/O.

Now, let $ e_1,e_2,\ldots , e_W$ be a topological order of the connections of the neural network (that is, whenever $e_i$ and $e_j$ are connections for which the output-neuron of $e_i$ is the input-neuron of $e_j$, we have $i<j$). Together with an eviction policy, a topological order of the connections gives rise to an inference-computation in a natural way, shown by Algorithm \ref{alginf}.

\begin{algorithm}[H]
\label{alginf}
\begin{algorithmic}[1]
\STATE \textbf{Input}: A topological order of the connections $e_1,\ldots,e_W$
\STATE \textbf{Output}: Values of the output neurons
\FOR{i = 1 \textbf{to} $W$}
\STATE Read the connection $e_i=(a,b,w)$;
\IF{value of the input neuron, $n_a$, is not in fast memory}
 \STATE Read $n_a$ (possibly, first evicting one value, if fast memory is full);
\ENDIF
\IF{partial sum of the output neuron, $n_b$, is not in fast memory}
\STATE Read $n_b$ (possibly, first evicting one value, if fast memory is full);
\ENDIF
\STATE Update $n_b=n_b+w\cdot n_a$

\IF{there is no connection after $e_i$ with output-neuron $n_b$}
    \STATE Apply the activation function: $n_b=f(n_b)$;
\ENDIF
    
\ENDFOR
\end{algorithmic}
\caption{Inference Algorithm}\label{alginf}
\end{algorithm}

Notice that this pseudocode generalizes all ``standard ways'' of performing inference. For example, matrix-vector-multiplication based inference would correspond to orders that start with all connections from the first layer, followed by the connections of the second layer and so on. Yet, it also allows us to perform inference on FFNN-architectures given by any possible DAG (including those with very ``chaotic'' skip connections) and not just those that are layered.

As for optimizing I/Os, it gives us more flexibility, as it allows us to employ any possible topological order of the connections, including those \emph{that do not correspond to layer-after-layer computations}. For example, it allows us to start computing neurons of the layers $i+1,i+2,\ldots$ even when not all neurons of the $i$-th layer have been computed. This can save I/Os, because when we finish computing a neuron from the $i$-th layer, we can directly reuse it to start computing neurons of the $(i+1)$-st layer, instead of storing it and reading it again later. 
Since for each connection we perform a computation for which we need the value of the input neuron and the partial sum at its output neuron, we would like to find an order in which connections that have neurons in common, are clustered together.

\section{Bounds for inference}
In this section we present generic bounds on the I/O-complexities of inference. They serve as guidelines for the more instance-specific optimizations that we consider later on.

\begin{Theorem}\label{oinft}
Let $\mathcal{N}$ be a connected FFNN and assume $M\geq 3$. Then the optimal number of I/Os for inference satisfies \begin{equation} \label{tbi}
        W+N+S\leq \textbf{I/Os}(\mathcal{N},M)  \leq 2\cdot(W+N-I).
    \end{equation}
     The optimal number of read-I/Os for this problem satisfies \begin{equation} \label{rbi}
        W+N\leq \textbf{rI/Os}(\mathcal{N},M)  \leq 2\cdot W +N-I.
    \end{equation}
     The optimal number of write-I/Os for this problem satisfies \begin{equation} \label{wbi}
        S\leq     \textbf{wI/Os}(\mathcal{N},M)  \leq N-I.
    \end{equation}
\end{Theorem}

\begin{proof}
The weights, biases and input values have a total size of $N+W$. Since we cannot perform inference without having read all of these data at least once, we obtain the lower bound for the number of read-I/Os. Likewise, we cannot perform inference with less than $S$ write-I/Os, because (by definition of the inference problem) we need to write the values of all $S$ output-neurons. Hence, the lower bound for the total number of I/Os follows by adding the lower bounds for the read- and write-I/Os.

Now it remains to show that we can do inference using no more than $N+2\cdot W-I$ read-I/Os, $N-I$ write-I/Os and $2\cdot(N+W-I)$ overall I/Os.
To achieve this, we fix a topological order of the non-input neurons: $n_1,\ldots,n_{N-I}$. Then, we reorder the connections in such a way that their output neurons appear in the order of this topological order (notice that this is also a topological order of the connections; see Figure~\ref{Orderings_both} for an illustration of this association of topological orderings). Notice that this order is naturally partitioned into intervals: It begins with an interval of connections ending in $n_1$, followed by the connections ending in $n_2$ and so on. 

We now show that performing inference in this order with our inference-algorithm (Algorithm \ref{alginf} from the main paper) and a MIN eviction policy, costs at most $2W+N-I$ read-I/Os and $N-I$ write-I/Os. As we start reading the connections, we spend $1$ read-I/O to read the bias of $n_1$ and then at most $2$ read-I/Os for each of the connections that end on $n_1$ ($1$ for reading the input neuron and $1$ for reading the connection itself). Once we passed through this interval of connections that end on $n_1$, we apply the activation to finish the computation of neuron $n_1$. Then we continue with the interval of connections that end on $n_2$. Also in this interval of connections as well as in all following intervals of connections, we spend $1$ read-I/O for the bias of the output-neuron and at most $2$ on each of the connections. Since there is exactly one interval of connections for each of the $N-I$ non-input neurons, we spend at most $N-I$ read-I/Os to read the biases and at most additional $2W$ for the connections and their input-neurons, adding up to $2W+N-I$ read-I/Os. 

Now we count the write-I/Os deployed in this computation. 
Since all connections that end in the same neuron follow each other, we always compute neurons without having to write temporary values (once we start computing a neuron all of the following computation steps are also directed towards computing this neuron, and we only start computing another neuron once the previous is finished). From this it follows, that if we spend a write-I/O on some neuron, then we are writing a fully computed neuron value to slow memory. And hence, if we read and evict this value again, we do not spend another write-I/O on this value (since the fully computed value is already stored in slow memory, the efficient eviction policy will evict this value by deleting it). Therefore, we spend at most one write-I/O for each of the non-input neurons. Since there are $N-I$ non-input neurons, we have overall at most $N-I$ write-I/Os.
Adding the read- and write-I/Os, we conclude that the total number of I/Os is at most $N+2\cdot W-I+N-I=2\cdot(N+W-I)$.

\end{proof}

\begin{figure}
\centering
\includegraphics[width=\linewidth]{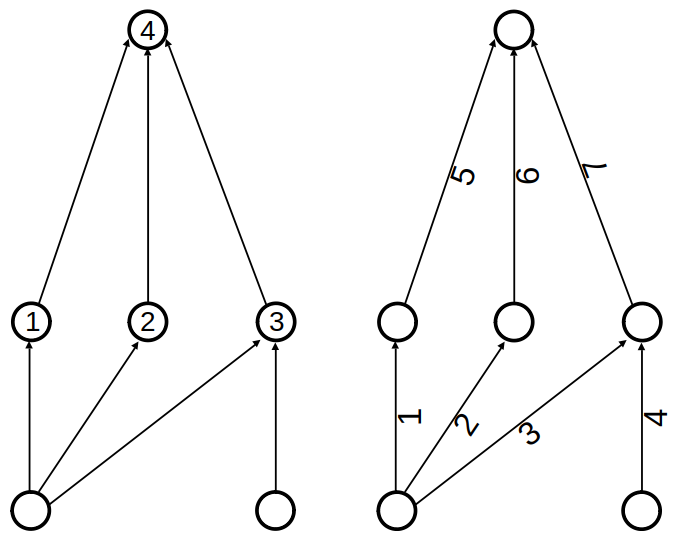}
\caption{Associated topological orderings of neurons and connections.}
\label{Orderings_both}
\end{figure}

Notice that in both Inequalities \ref{tbi} and \ref{rbi}, the term on the right hand side is at most twice as large as the term on the left hand side. The next Proposition establishes that none of the generic bounds given by this theorem could be tightened any further by multiplying it with a constant other than $1$.

\begin{Proposition}\label{tight}
As generic bounds that depend only on $W,N,I$, and $S$, the bounds in Theorem \ref{oinft} are tight: For each one of them and for any $\epsilon>0$, there are instances for which the true value differs from the bound by a multiplicative factor that lies in $[1-\epsilon,1+\epsilon]$.
\end{Proposition}

To prove this Proposition, we need to show that for each of the bounds of Theorem \ref{oinft}, there are instances that are arbitrarily close to the bound (close in terms of multiplicative factors).

The first lemma establishes that there are instances that exactly attain all lower bounds of Proposition \ref{tight} (and hence, are obviously arbitrarily close).

\begin{Lemma}\label{tightd}
On any FFNN $\mathcal{N}=[L_1,L_2,\ldots,L_d]$ in which any two consecutive layers $L_i,L_{i+1}$ have together at most $M-1$ neurons, we can do inference, with a number of read-, write-, and total I/Os that matches the lower bounds given by Theorem \ref{oinft}.
\end{Lemma}
\begin{proof}
We prove this by describing a computation strategy that achieves this.

Read the $|L_1|$ input values into fast memory. Also initialize in fast memory $|L_2|$ partial sums with the biases of the neurons in $L_2$. Since $|L_1|+|L_2|\leq M-1$, we have in fast memory free space for at least one more value. We use this free space to iterate over the weights between $L_1$ and $L_2$ and perform the following steps for each one of them:
\begin{enumerate}
    \item Read the weight.
    \item Multiply it with the value of its input neuron (which is already in fast memory).
    \item Add this product to the partial sum of its output neuron.
    \item Delete the weight from fast memory.
\end{enumerate}
Once this iteration is finished, we apply the activation function to the partial sums of $L_2$, which finishes the computation of the neurons in this layer. Now we do not need the values from $L_1$ anymore and can delete them from fast memory. This gives us enough free space to read the biases of $L_3$, and compute this layer in the same way. We proceed like this, computing layer after layer. Once we finish the computation of the neuron-vales of the last layer, we write those $S$ values back to slow memory.

Notice that we read each of the $N-I$ biases, $I$ input values, and $W$ weights exactly once. Hence we use $N+W$ read-I/Os. We only use write-I/Os to write the $S$ output values. Hence, altogether we use $N+W+S$ I/Os.
\end{proof}

Alternatively, we could have proved the previous Lemma by verifying that an FFNN with the mentioned properties, can be constructed with the Compact Growth method that we introduce later.
The next lemma shows that there are instances that use a number of read- and total I/Os that are arbitrarily close to their upper bound.
\begin{Lemma}\label{tightur}
For every $\epsilon>0$, there exist FFNNs $\mathcal{N}$ (of arbitrary size) that have $I/Os(\mathcal{N},M)>(1-\epsilon)\cdot 2\cdot (W+N-I)$ and $rI/Os(\mathcal{N},M)>(1-\epsilon)\cdot (2\cdot W+N-I)$.
\end{Lemma}
\begin{proof}
This is the case for architectures, where a large proportion of the neurons are input neurons. For example, in a tree with $I$ input neurons, all connected to a single output neuron, we have $I/Os(\mathcal{N},M)=\cdot 2\cdot (W+N-I)$ and $rI/Os(\mathcal{N},M)= 2\cdot W+N-I$.
\end{proof}

The next lemma shows that there are instances that use a number of write-I/Os that is arbitrarily close to the upper bound.

\begin{Lemma}\label{tightuw}
For every $\epsilon>0$, there exist FFNNs $\mathcal{N}$ (of arbitrary size) that have $wI/Os(\mathcal{N},M)>(1-\epsilon)\cdot (N-I)$.
\end{Lemma}
\begin{proof}
Notice that this inequality is satisfied by any FFNN for which all neurons are either input or output neurons. But to give less trivial examples, we will construct FFNNs with hidden neurons that satisfy this inequality.

For $M,h,w\in\mathbb{N}^+$, consider the FFNN that has $I$ input neurons, $S$ output neurons, and one hidden layer with $h$ neurons. Clearly, this FFNN requires at least $S$ write-I/Os. So, for any parameter configuration $I,h,S$ with $S>h\cdot (1-\epsilon)/\epsilon$, we have $wI/Os(\mathcal{N},M)>(1-\epsilon)\cdot (N-I)$.
\end{proof}

The three previous lemmata provide the \emph{tight extremal cases} that bring us into the position to prove Proposition \ref{tight}. Although the constructed architectures in the proofs of these lemmata were ``artificial'', they suffice to show that generic bounds that only depend on $I,S,N,$ and $W$, cannot be tighter (tighter by a constant multiplicative factor other than $1$) than our bounds. Yet, notice that, despite being ``artificial'', these extremal cases are more than just ``a few cornercases''. In fact, the constructions in each of these proofs could have been chosen to be arbitrarily large. 

\begin{proof}[Proof of Proposition \ref{tight}]
Lemma \ref{tightd} shows that the lower bound on the number of read-, write-, and total I/Os is tight.
Lemma \ref{tightur} shows that the upper bound on the number of read-I/Os and total I/Os is tight. Lemma \ref{tightuw} shows that the upper bound on the number of write-I/Os is tight.
\end{proof}

Notice that none of the bounds depends on $M$. 
This is surprising, because for most computational problems, I/O-complexities depend strongly on the memory size\footnote{For example, in the case of matrix multiplication, the dependency on $M$ is of the order $O(1/\sqrt{M})$ \cite{redblue}.}. As we increase $M$, we gain more flexibility that we can use to avoid I/Os. Also surprising is the fact that these bounds do not depend on specific properties of the network, except for its sizes. This \emph{does not mean that the overall I/O-complexity does not depend on these quantities} (memory size and  FFNN architecture), but from the tightness of our bounds we conclude that the dependence is only moderate.

For write-I/Os, this is very different.
The bounds for write-I/Os given by Theorem \ref{oinft} are less tight. In fact, they can be arbitrarily loose (despite being optimal in the sense of Proposition \ref{tight}). The problem for write-I/Os is that, it is simply not possible to give interesting bounds that only depend on $W,N,I$, and $S$ but no other factors such as the connection order, the memory size $M$, and the FFNN architecture (this impossibility is implied by Proposition \ref{tight}). 
The following Proposition shows that inference in a \emph{layer-after-layer} fashion can be arbitrarily more expensive than using an optimal order. 

\begin{Proposition} \label{nogoodlayer}
    For every $c\in\mathbb{N}$ and every memory size $M\in\mathbb{N}$, there exists an FFNN $\mathcal{N}_{M,c}$, such that performing inference in a \emph{layer-after-layer} fashion, requires at least $c$ times more write-I/Os than optimal.
\end{Proposition}
\begin{proof}
Consider a sparse FFNN that has $c+2$ layers: $1$ neuron in the input layer, $2M$ neurons in each of the $c$ hidden layers, and one neuron in the output layer. Let each hidden neuron have exactly one incoming and one outgoing connection, let the input-neuron be connected to each neuron of the first hidden layer, and let the output-neuron have incoming connections from all neurons of the last hidden layer (in other words, this FFNN has $2M$ chains of neurons of length $c+2$ that meet in the input and output neuron).
If we perform inference on this FFNN in a layer-after-layer fashion, we need the values of all $2M$ neurons of each hidden layer to compute the next layer. Since our fast memory has only capacity for $M$ neurons, we will need to store the other $M$ neurons, requiring at least $M$ writes for each hidden layer.
Hence, if we perform inference on this FFNN in a layer-after-layer fashion, we would need to use at least $M\times c$ write-I/Os. Yet, if we compute chain after chain, a single write-I/O would suffice.
\end{proof}

\section{Connection Reordering: Adapting the order to the FFNN and hardware}
In this section we introduce \emph{Connection Reordering}, which is a method to optimize the topological order of the connections for a given FFNN architecture and memory size $M$. This method depends on the following hyperparameters: The \textbf{window size} $ws\in\mathbb{N}$, the \textbf{cooling rate} $\sigma\in \mathbb{R}$, and the \textbf{number of iterations} $T\in\mathbb{N}$. The high-level idea of this method is based on Simulated Annealing~\cite{SA}: Over $T$ iterations, we perform random changes to the topological order (we call this \textit{creating neighbors}) and either retain or discard the changes with a probability that depends on the quality of the old and the new order (called \textit{updating}). Now, we fully specify this method by describing the processes of creating neighbors and updating.  

\subsection{Creating neighbors}
We start with a topological order of the connections $e_1,e_2,\ldots, e_W$ and let \textbf{OldI/Os} $\in \mathbb{N}$ denote the number of I/Os used in this topological order. The number of I/Os obviously depends on the memory size and eviction policy, but those are fixed throughout the execution of this algorithm and hence the number of I/Os depends only on the topological order of the connections. 

First, we choose uniformly at random one connection $e_i$. We let $w$ be an integer that we choose uniformly at random from $\lbrace 0,1,\ldots, ws-1\rbrace$. We consider the window of connections $e_i,e_{i+1},\ldots ,e_{\min(i+w,W)}$. We choose the direction in which we move the connections from this window: Either left or right with probability $0.5$. 

\paragraph{Case 1: Moving to the left}
If we move the connections to the left, we start moving the leftmost connection $e_i$ of the window. We move $e_i$ to the left until we encounter another connection $e_s$ that has the same input neuron as $e_i$, or whose output neuron is equal to the input neuron of $e_i$ (if we never encounter such a connection, we move $e_i$ to the very beginning of the order). We insert $e_i$ right next to $e_s$ so that $\ldots,e_s,e_i,e_{s+1},\ldots$ is the new order. Note that this is again a \emph{topological} order. Then we continue moving the second leftmost connection in the window in the same fashion. We continue like this until we moved all connections from this window.

\paragraph{Case 2: Moving to the right}
If we move the connections to the right, we start moving the \emph{rightmost} connection of the window. We move it until we encounter another connection $e_z$, which has the same \emph{output} neuron as $e_i$ or whose input neuron is equal to the output neuron of $e_i$. We insert $e_i$ right before $e_z$ so that $\ldots,e_{z-1},e_i,e_{z},\ldots$ is the new order. Then we continue with the second rightmost connection from the window. We continue like this until we moved all connections from this window. 

\subsection{Updating}

After we moved all connections from the window, we have a new topological order. We denote it $\tilde{e_1},\tilde{e_2},\ldots,\tilde{e_W}$. We measure how many I/Os are used to perform inference in this new order and call this number \textbf{newI/Os}. If \textbf{newI/Os}<\textbf{oldI/Os}, we certainly update $e_1=\tilde{e_1},e_2=\tilde{e_2},\ldots,e_W=\tilde{e_W}$ and \textbf{oldI/Os}$=$\textbf{newI/Os}. 

If \textbf{newI/Os}$\geq$\textbf{oldI/Os}, we either update or go back to the old order. We decide this at random, choosing to update with a probability 
    $2^{-(\text{newI/Os}-\text{oldI/Os})\cdot t^{\sigma}},$
where $t$ is the iteration number.

\section{Compact Growth: Adapting FFNNs to hardware}\label{genFNNs}

In this section we present a construction scheme that completely characterizes the FFNN architectures that allow inference with a minimal number of I/Os for a given memory size $M$. This answers the following question: For a given memory size $M$, which are the FFNN architectures on which we can do inference with a number of I/Os that matches the lower bound given by Theorem \ref{oinft}?  Or equivalently: Which are the FFNNs on which we can do inference without having to write or read intermediate values due to lack of memory?
The idea is to couple the construction of the FFNN closely to the steps in the inference computation. More precisely, we will build the FFNN by a sequence of steps of four different types.

We introduce the concepts of \textit{pebbles} and \textit{bags} to reason about I/Os in our computation. Each pebble corresponds to one neuron. A pebble can be either gray (if it is not yet fully computed)
or black (if it is fully computed and can be used by its outgoing connections). The bag represents the fast memory.

\subsection{The construction rules}

We start our construction with an ``empty FFNN'' (there are no neurons and no connections) and an empty bag that represents the content of our fast memory. We build up the FFNN by a sequence of construction steps of the following four types (for each type we write in parentheses the corresponding computation step into which it can be translated).

\textbf{1)} Whenever we have less or equal than $M-2$ pebbles in our bag, we can either add a gray or a black pebble to our bag and simultaneously add an isolated neuron to our FFNN (reading a neuron into fast memory).\\
\textbf{2)} When we have a black and a gray pebble in our bag, we can draw a connection in our FFNN, that starts at the neuron that corresponds to the black pebble and ends at the neuron that corresponds to the gray pebble. (multiply weight with the value of the input neuron and add it to the partial sum of the output neuron).\\
\textbf{3)} We can turn a gray pebble into a black pebble (finish computation of the neuron by applying activation function).\\
\textbf{4)} We can remove a black pebble from the bag (delete from fast memory).\\

The following theorem summarizes the properties of FFNNs built by a sequence of the above steps.

\begin{Theorem}[Compact growth]\label{compactGrowth}
Let $\mathcal{N}$ be a connected FFNN with $W$ weights and $N$ neurons. Given a memory of size $M\geq 3$, we can do inference on $\mathcal{N}$ with $N+W$ read-I/Os and $S$ write-I/Os if and only if $\mathcal{N}$ can be constructed by the compact growth scheme that we just described.
\end{Theorem}

\begin{proof}[Proof of Theorem \ref{compactGrowth}]
If we can build $\mathcal{N}$ by a sequence of these construction steps, then the corresponding steps in the parentheses, describe a valid sequence of computation steps for inference. By only allowing to insert pebbles when we have (strictly) less than $M-1$ pebbles in the bag, we ensure that there are never more than $M-1$ pebbles in memory and hence there is one free space for a connection. Hence, this sequence of computation steps can be executed with a fast memory of size $M$ without requiring any I/Os on temporary values.
This proves the ``if'' part of the statement.

Conversely, assume for a given $\mathcal{N}$ and $M\geq 3$, we can order the connections in a way that allows us to perform inference with a minimal number of I/Os.
Denote this order $e_1,e_2,\ldots,e_W$. We will show how we can translate the computation steps into a sequence of pebble-moves that creates the FFNN.
According to Theorem \ref{oinft}, a \textit{minimal number of I/Os} means that we perform $W+N$ read-I/Os and $S$ write-I/Os. Since this is the number of read-I/Os required to read all necessary parameters once, and write out the values of the output-neurons once, these I/Os only suffice to perform these necessary reads and writes. In particular, we cannot spend any more I/Os to write or read temporary values. 
Hence, for each neuron we will read one value at some moment (the bias or the input value), then have this value in fast memory for a certain interval of time during which we update it according to the progress of the computation, and then evict the fully computed neuron value when it is not needed anymore. 
After that we will not be able to use this value again. We translate this into our pebble construction: Putting a gray pebble into our bag, having it there for a certain amount of time, transforming it into a black pebble (once the neuron is fully computed), and then removing it (when the neuron value is not needed anymore and is evicted).

During this interval, this optimal computation strategy must use all incoming and outgoing connections of this neuron: The incoming connections during the first part of the interval (when the neuron is not yet fully computed and hence the pebble is gray) and the outgoing connections during the second part. Further, any time we use a connection, the ``other end'' of the connection must also be in fast memory and hence the corresponding pebble must be in the bag (remember that we insert a pebble whenever we read a bias/input and remove the pebble when we evict the value). Also, any time we use a connection, the color of the pebble of the input neuron is already black (because we turn a pebble black as soon as the neuron is fully computed) and the pebble of the output neuron is still gray. Hence, as we use the connections on the computation side, we can draw the connections on the side of the pebble model.

Due to the restricted memory size, we can assume without loss of generality that we have never more than $M-1$ neuron values in fast memory (if we read an $M$-th value, there is no more space for a connection, and the computation is stuck until we remove one neuron value; but then we could just as well first remove this neuron value before reading another one). And hence, we can assume without loss of generality that we only read values when there are at most $M-2$ values in fast memory. On the side of the pebble construction, this ensures that we only add pebbles, when there are at most $M-2$ pebbles in the bag.
\end{proof}

We can apply this theorem to the problem of \emph{optimal hardware for a given FFNN}. To do this, we consider the \textbf{bandwidth} of an FFNN, defined as the smallest $k\in\mathbb{N}$ for which there exists a topological order of the neurons, such that any two connected neurons are at most $k$ steps apart in the topological order. We prove now that when FFNNs have bandwidth $k$, we can build them with compact growth, using $M=k+2$. 
\begin{Corollary}\label{bwic}
    Let $\mathcal{N}$ have bandwidth $k\in \mathbb{N}$. If we have a memory size $M\geq k+2$, then we can perform inference on $\mathcal{N}$ without reading or writing any temporary values. 
\end{Corollary}

\begin{proof}
Let $\mathcal{N}$ have bandwidth $k$. This means that there exists a topological order of the neurons $n_1,n_2,\ldots ,n_N$, such that for any $i\in \lbrace 1,\ldots,N \rbrace$, all incoming connections to $n_i$ are among the $k$ previous neurons $\lbrace n_{max(1,i-k)},n_{max(1,i-k)+1},\ldots,n_{i-1}\rbrace$. So, we can build $\mathcal{N}$ with compact growth, by adding the pebbles according to the topological order of the neurons, and when the bag is full, removing the pebble that was added $k+1$ steps ago. This ensures that when we add a pebble of a neuron $n$ to the bag, the pebbles of all other neurons on which $n$ directly depends are also in the bag and we can add all connections that go to $n$. Hence, with a memory size of $k+2$, we can build $\mathcal{N}$.
\end{proof}

\begin{figure}[h!]
    \begin{subfigure}[b]{0.95\linewidth}
    \centering
    \includegraphics[width=\linewidth]{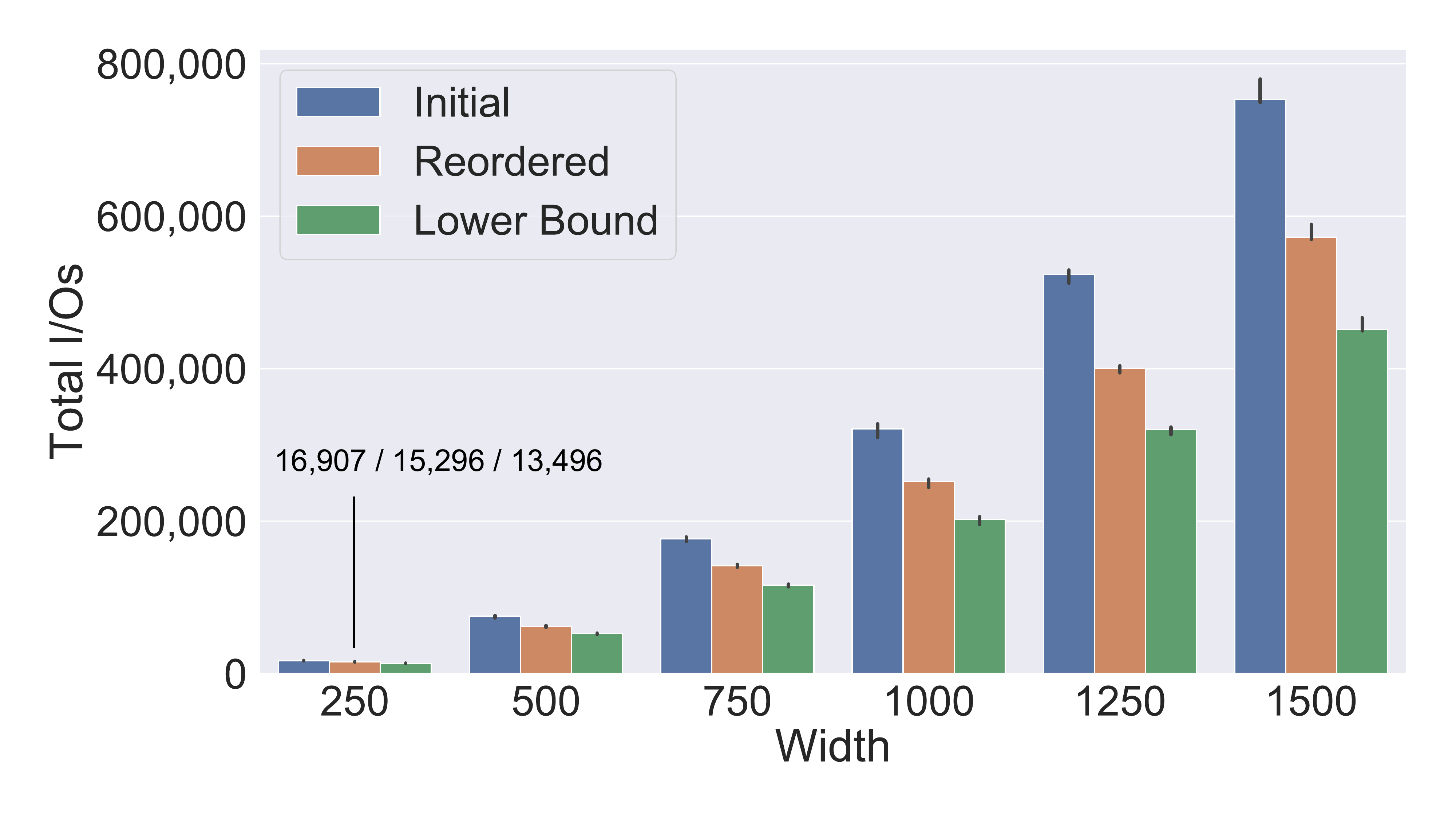}
    \vspace{-2em}
    \caption{Neurons per Layer}
    \end{subfigure}
    \begin{subfigure}[b]{0.95\linewidth}
    \centering
    \includegraphics[width=\linewidth]{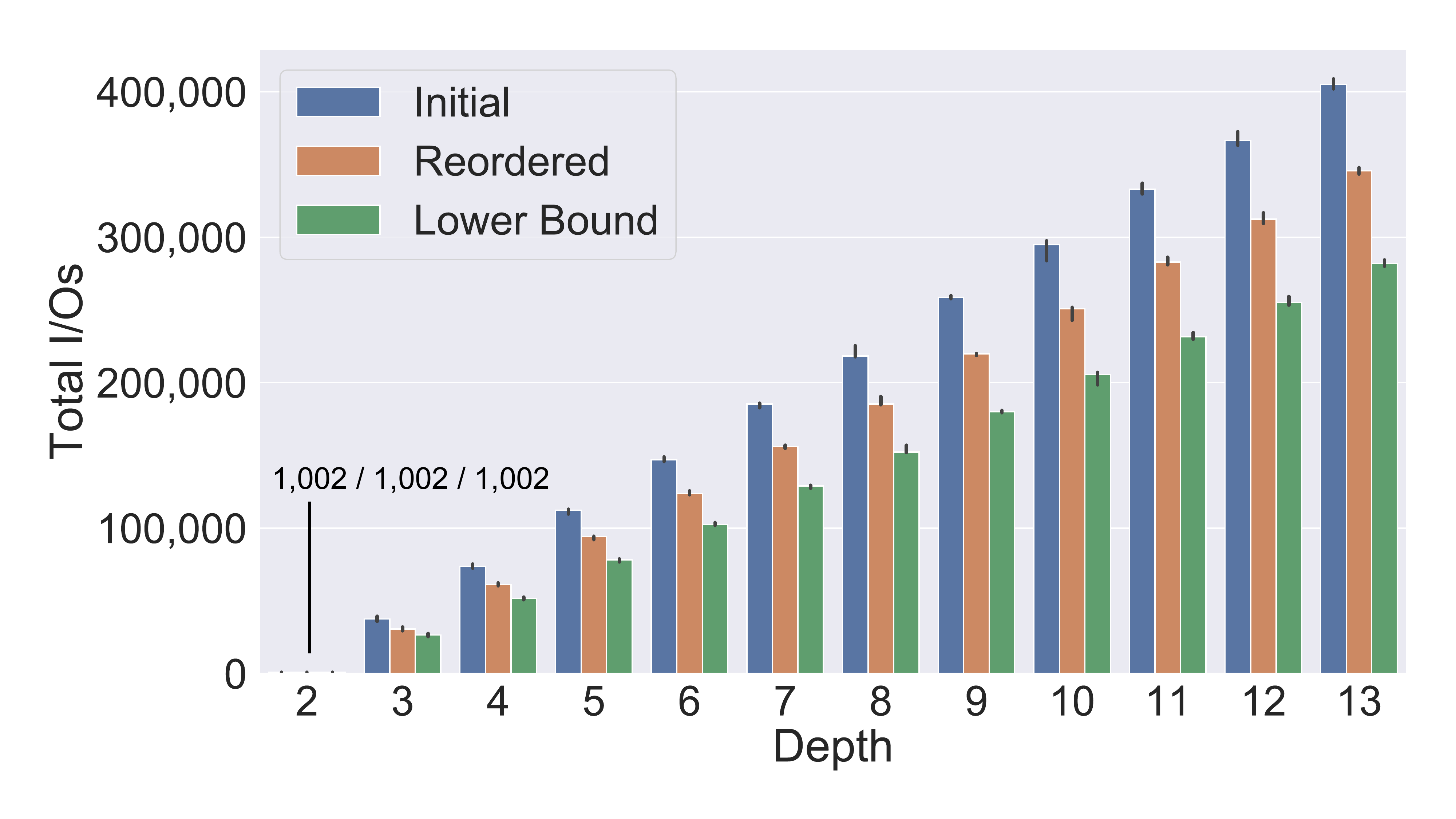}
    \vspace{-2em}
    \caption{Depth}
    \end{subfigure}
    \begin{subfigure}[b]{0.95\linewidth}
    \centering
    \includegraphics[width=\linewidth]{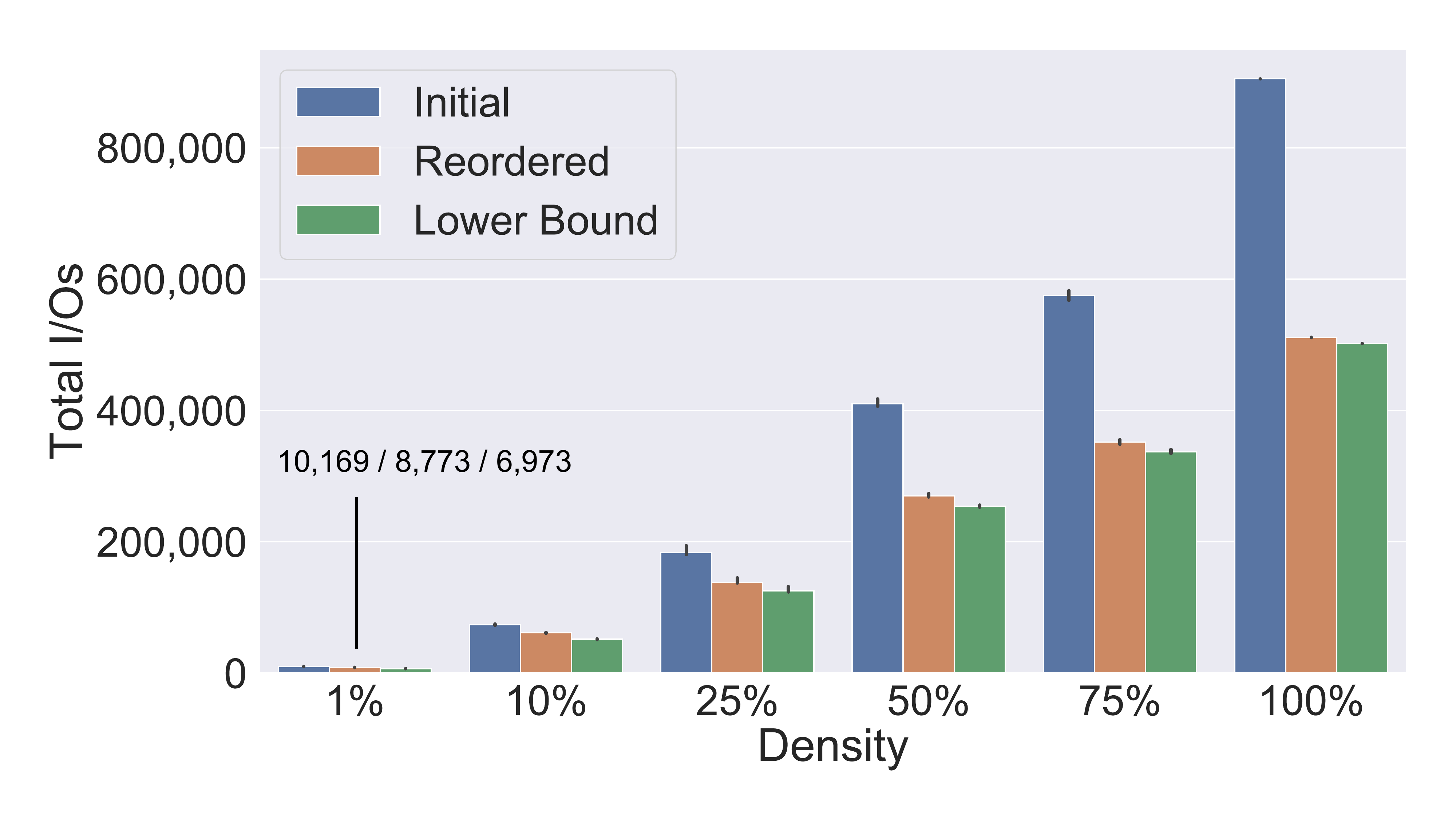}
    \vspace{-2em}
    \caption{Connection Density}
    \end{subfigure}
    \begin{subfigure}[b]{0.95\linewidth}
    \centering
    \includegraphics[width=\linewidth]{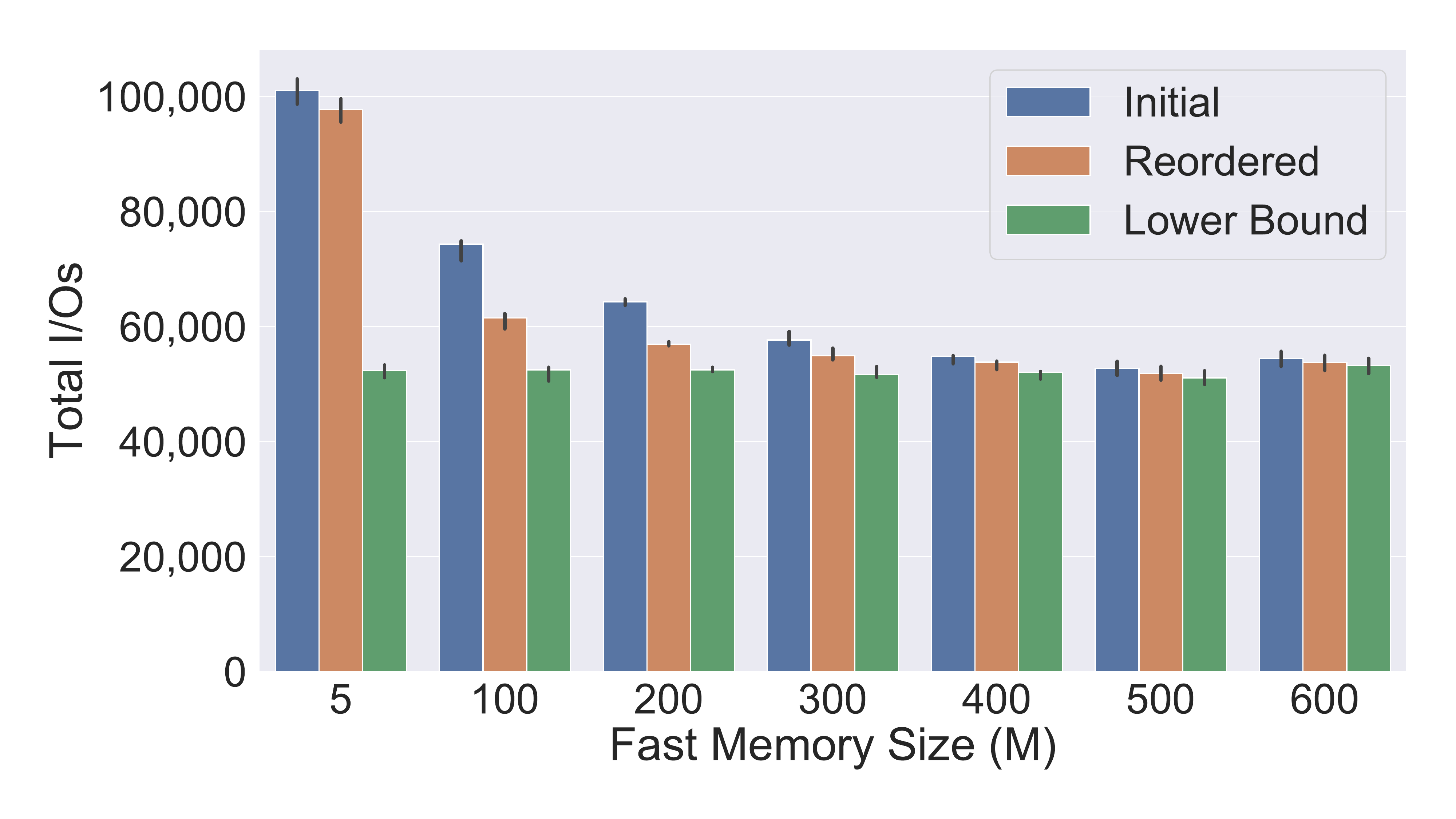}
    \vspace{-2em}
    \caption{Fast Memory Size}
    \end{subfigure}
    \caption{Connection Reordering for varying sparse neural network properties. Unless otherwise stated, the network is a 500-neuron wide, 4-layer MLP, with an edge density of 10\% and $M=100$. The lower bound is obtained from Theorem \ref{oinft}. When too small, I/Os for MLPs are annotated.}
    \label{fig:reordering-results}
    \vspace{-1em}
\end{figure}

\section{Experimental Evaluation}
To demonstrate both the theoretical accuracy as well as the practical utility of our results, we run both simulated experiments counting I/Os and measure CPU performance of the reordered sparse neural networks compared with high-performance libraries.

As inputs, we use two types of networks: MLPs with varying depth, breadth, and density; and fully-connected layers from the popular BERT~\cite{bert} Transformer~\cite{attention} neural network. For the former, we generate five randomly-sparse FFNNs (obtained by random edge sampling or using Compact Growth). For the latter, we use a pre-trained BERT$_\text{\textsc{LARGE}}$ neural network and perform magnitude pruning~\cite{sparsesurvey}.

\subsection{Simulated Experiments}

To test Connection Reordering (CR) and Compact Growth (CG), we implement Algorithm 1 and cache simulation, along with LRU, RR, and MIN eviction policies. 
We run each experiment configuration with five random MLPs and BERT, reporting the median values and 95\% nonparametric confidence intervals as error bars.

\subsubsection{Connection Reordering}
\label{Simux}
We evaluate CR by optimizing the I/Os of random sparse FFNNs with varying structural and sparsity properties. Figure~\ref{fig:reordering-results} shows four dimensions in which we vary a baseline FFNN: a 10\% dense, $4$-layer MLP with $500$ neurons in each layer and one output neuron, running inference with a fast memory size of $100$ elements. We either vary density, width, depth, or fast memory size, while keeping the other parameters constant at their baseline value. We use a MIN eviction policy. Initially, we order the connections as in the proof of Theorem \ref{oinft}, so that we are guaranteed to not use more than twice the optimal number of I/Os. Then, we apply CR with $T=$1,000,000, $\sigma=0.2$, and set $ws$ to four times the average in-degree of the network. As we can see in these examples, CR reduces the total number of I/Os further by up to $43.5\%$ and brings us up to $97.4\%$ closer to the theoretical lower bound. Further, we can see that CR gives consistently large improvements across all parameter configurations (except for the cases where the initial I/O-complexity is already close to the lower bound).

\begin{figure}[t]
    \centering
    \begin{subfigure}[b]{.8\linewidth}
    \centering
    \includegraphics[width=\linewidth]{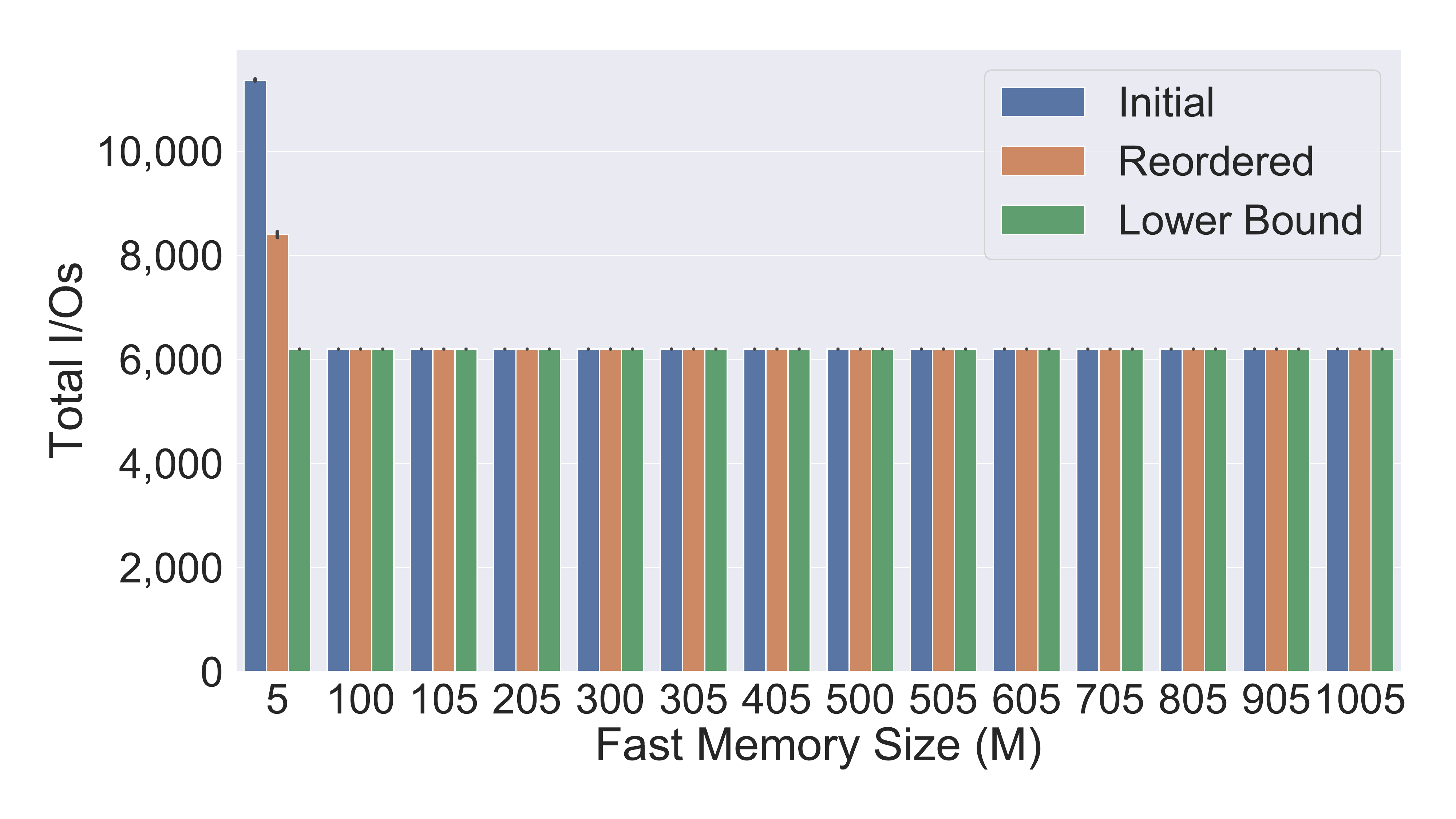}
    \vspace{-2em}
    \caption{$M_g=100$}
    \end{subfigure}
    \begin{subfigure}[b]{.8\linewidth}
    \centering
    \includegraphics[width=\linewidth]{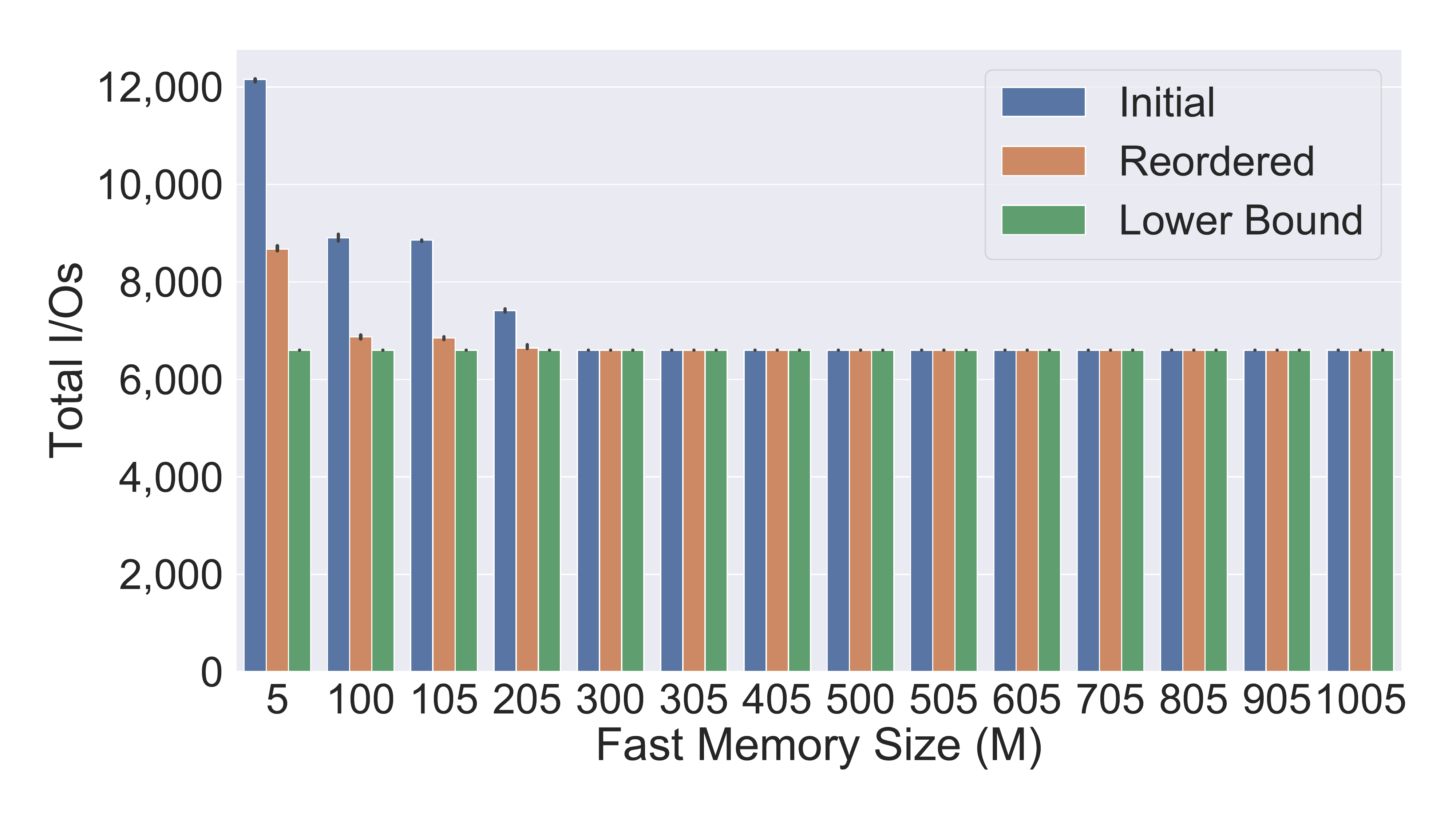}
    \vspace{-2em}
    \caption{$M_g=300$}
    \end{subfigure}
    \begin{subfigure}[b]{.8\linewidth}
    \centering
    \includegraphics[width=\linewidth]{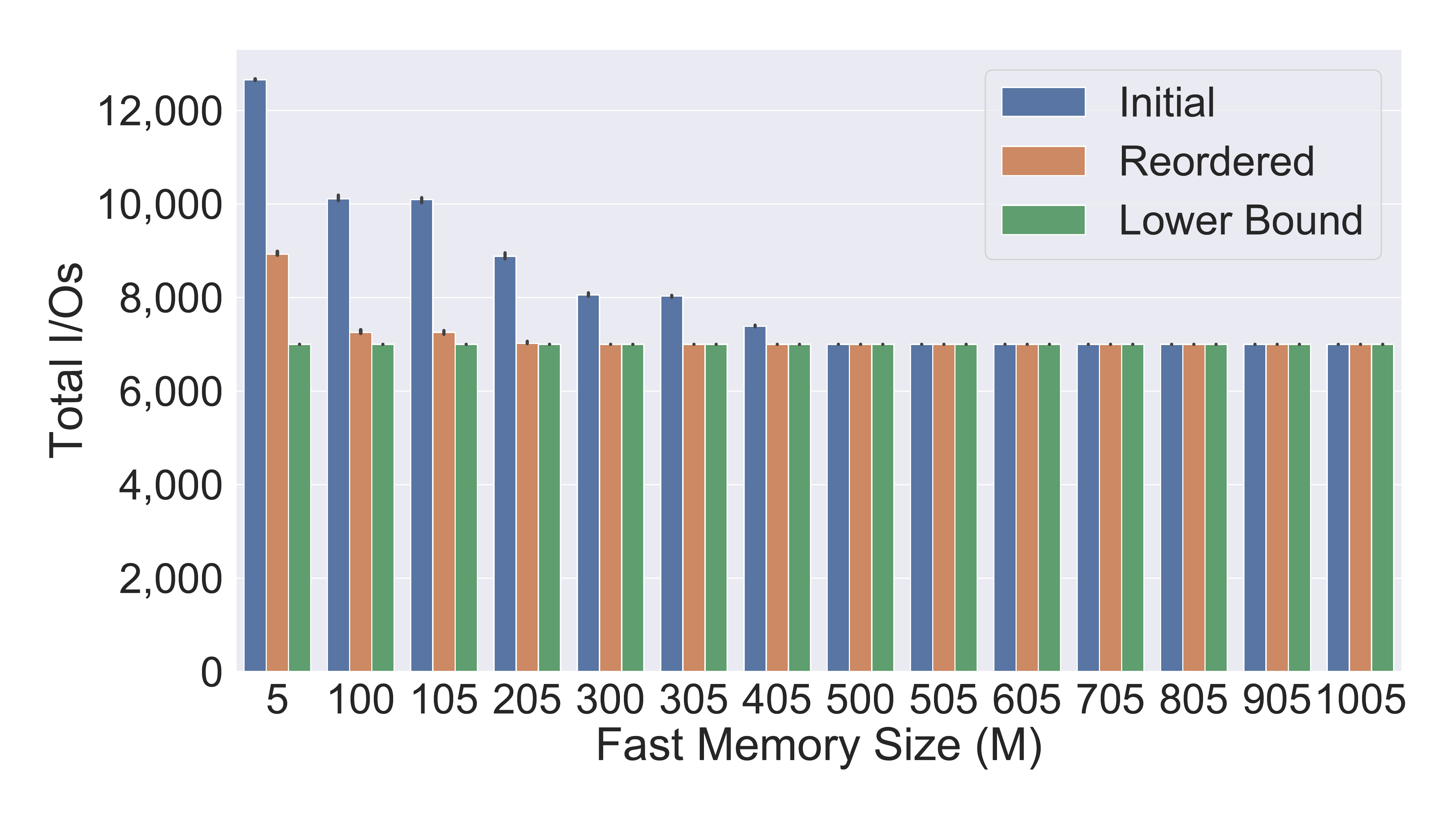}
    \vspace{-2em}
    \caption{$M_g=500$}
    \end{subfigure}
    \caption{Study of Compact Growth-generated FFNNs designed for different fast memory sizes $M_g$ ($1,000$ neurons with in-degree of $4$ and one output-neuron with in-degree $M_g$).}
    \label{fig:cgrowth-details}
\end{figure}

\subsubsection{Compact Growth}
We generate three FFNNs with Compact Growth as described in Appendix \ref{CGD}. These FFNNs correspond to the values $M_g=100,300$, and $500$. Then we count the I/Os used for these FFNNs with varying sizes of fast memory $M$. These numbers are shown in Figure \ref{fig:cgrowth-details}. As predicted by Theorem \ref{compactGrowth}, we can see for all three FFNNs that when $M\geq M_g$, we use a minimal number of I/Os. In these cases, we directly achieve this minimal number of I/Os using the connection order produced by CG, without applying CR. 

When we additionally apply CR, we are also able to use a minimal number of I/Os in some ranges of memory size $M$ with $M<M_g$. Of course this is not a contradiction to Theorem \ref{compactGrowth}: The theorem says that we can achieve a minimal number of I/Os using a memory size $M$ if and only if the FFNN can be built using $M_g=M$. But this does not imply that the same FFNN cannot be built with $M_g<M$.
\label{app:cg}

\subsubsection{Cache Eviction Policies}

\begin{figure}[h!]

    \centering

    \includegraphics[width=\linewidth]{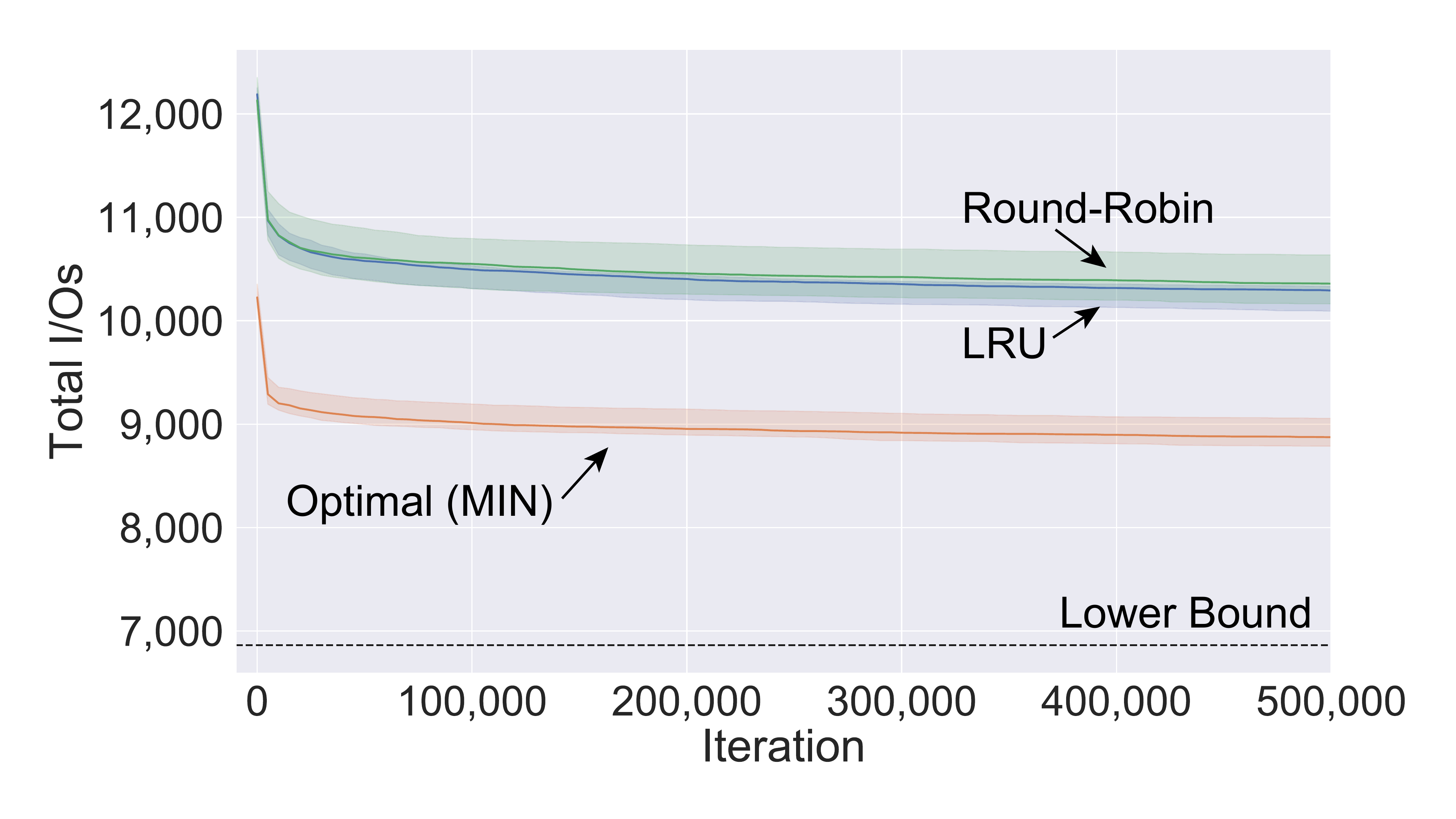}

    \caption{I/Os for different cache eviction policies.}

    \label{fig:policies}

\end{figure}

Figure~\ref{fig:policies} shows the I/O-complexity evolution (over Simulated Annealing iterations) of random FFNNs, for the RR, LRU, and MIN cache eviction policies. Owing to the update scheme, we observe a decaying convergence rate towards a minimum value, where the majority of I/O reduction happens in the first 10,000 iterations. It is interesting to observe that upon optimization, both RR and LRU converge to similar I/Os, which indicates that CR can tune FFNNs for hardware architectures with specific caching mechanisms.

\subsubsection{Size of Fast Memory}\label{app:fastmem}
In Figure~\ref{fig:fastmem}, we show how the number of I/Os changes as we increase the size of the fast memory, both before and after Connection Reordering (CR). As we can see, once the fast memory is sufficiently large, we use a number of I/Os that matches the theoretical lower bound provided by Theorem \ref{oinft}, with and without CR. But for insufficient memory size, we see that with CR we reduce I/Os and converge faster to the lower bound.
\begin{figure}[t]
    \centering
    \includegraphics[width=\linewidth]{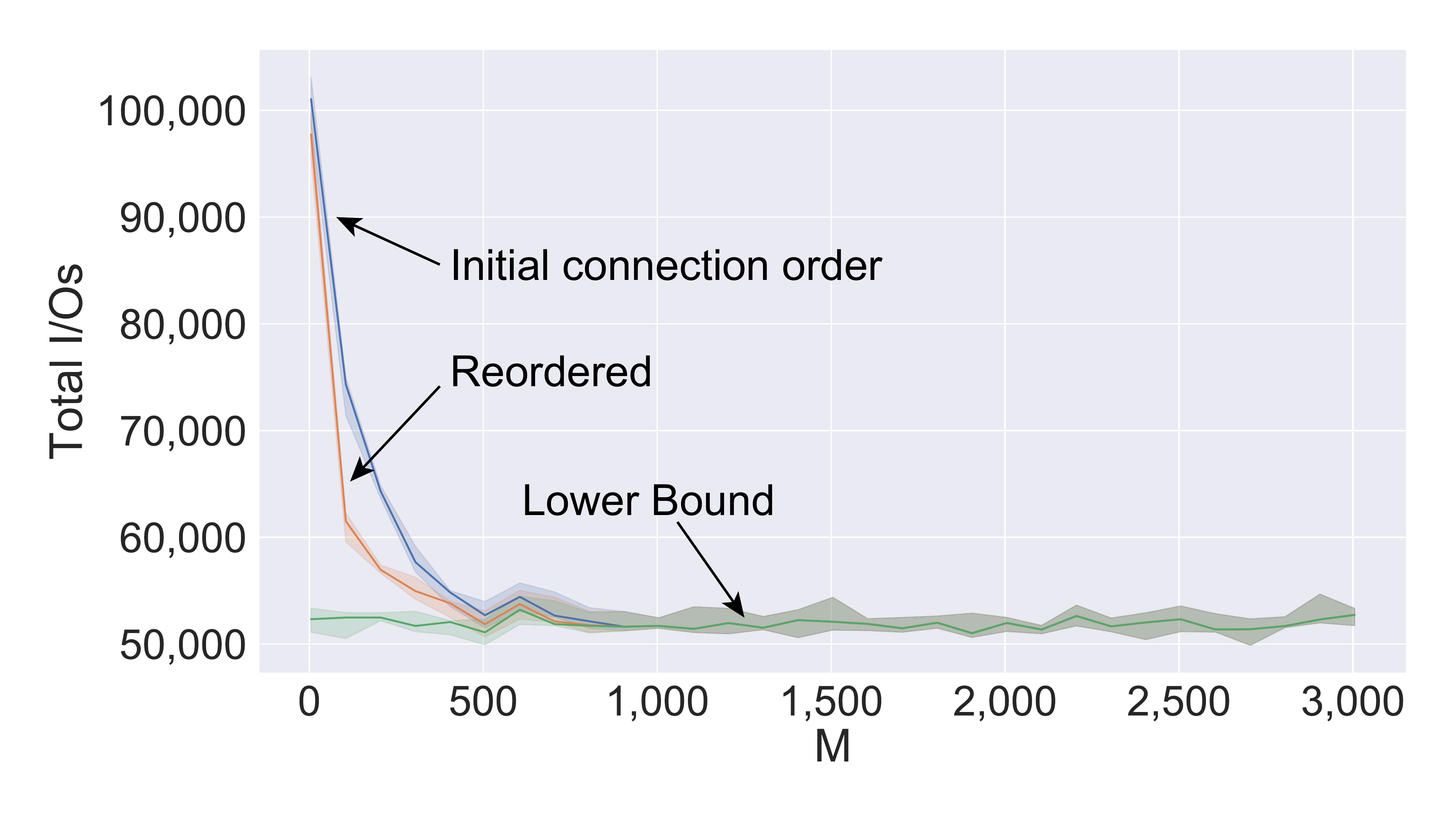}
    \caption{Study of different fast memory sizes ($M$) on random sparse FFNNs ($3$ layers of $500$ neurons each followed by one output neuron, $1$\% density).}
    \label{fig:fastmem}
\end{figure}

\subsubsection{Application to Transformer-model BERT}

\begin{figure}[t]
    \centering
    \begin{subfigure}[b]{0.95\linewidth}
    \centering
    \includegraphics[width=\linewidth]{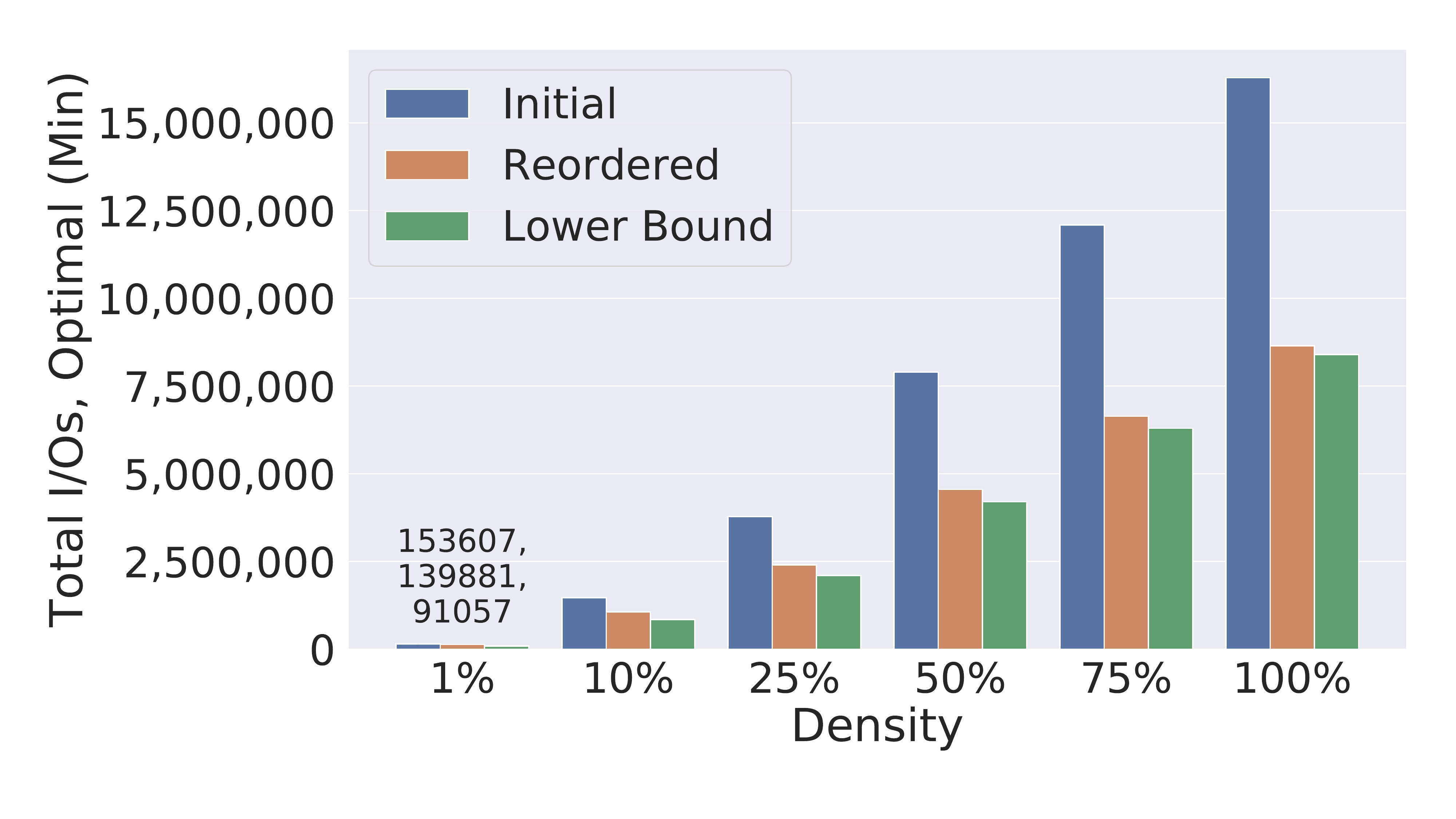}
    \vspace{-2em}
    \caption{Optimal (MIN)}
    \end{subfigure}
    \begin{subfigure}[b]{0.95\linewidth}
    \centering
    \includegraphics[width=\linewidth]{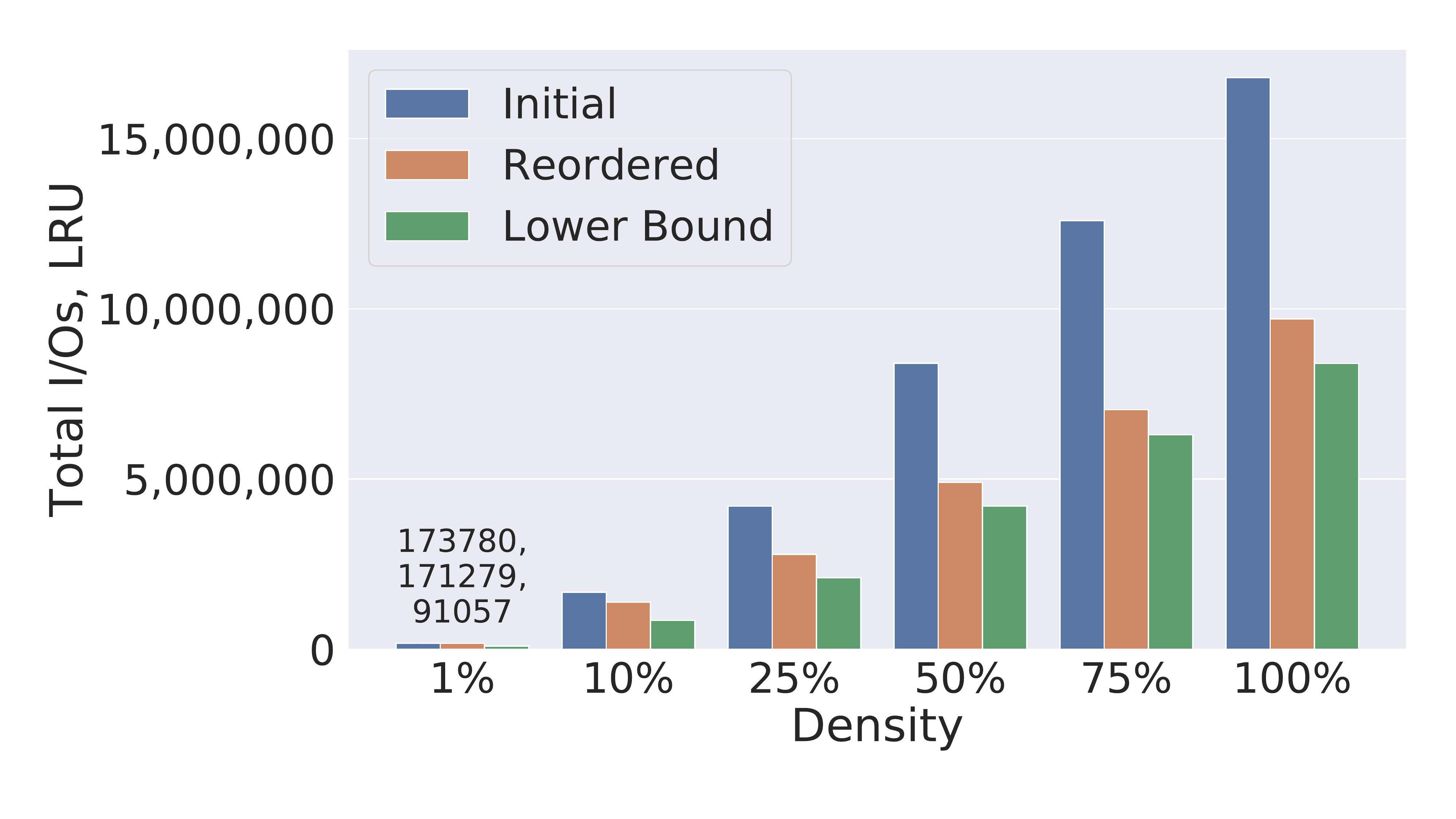}
    \vspace{-2em}
    \caption{LRU}
    \end{subfigure}
    \begin{subfigure}[b]{0.95\linewidth}
    \centering
    \includegraphics[width=\linewidth]{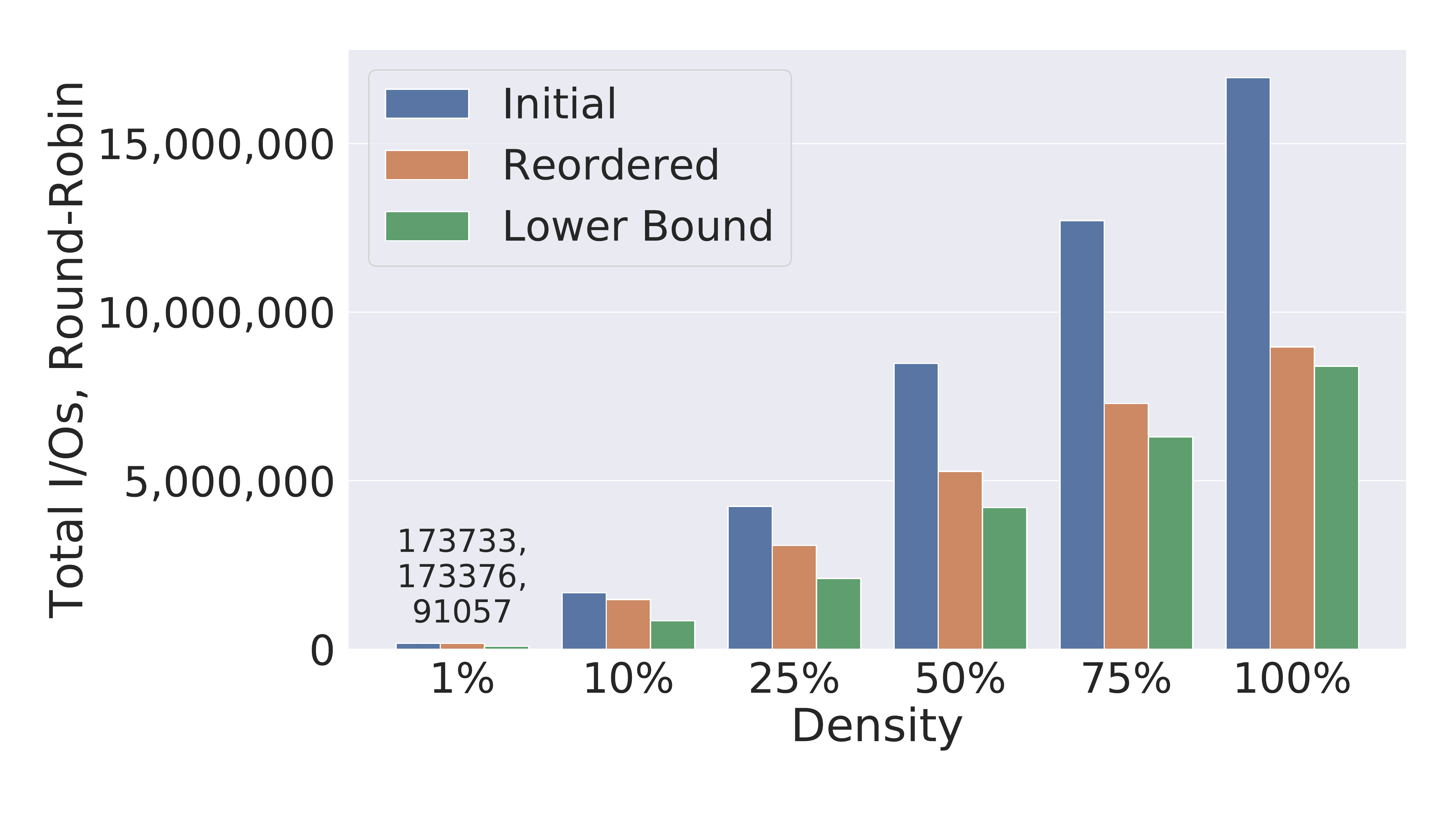}
    \vspace{-2em}
    \caption{Round-Robin}
    \end{subfigure}
    
    \caption{Connection Reordering for an MLP from BERT$_\text{\textsc{LARGE}}$ using different eviction policies, different edge densities and $M=100$. The lower bound is obtained from Theorem \ref{oinft}.}
    \label{BERT}
    \vspace{-1em}
\end{figure}

Transformers~\cite{attention} are a DNN architecture built from so-called \textit{encoders} and \textit{decoders}. A characteristic ingredient of transformers are \textit{attention heads}, which make them suitable for natural language processing and related tasks. A prominent example of a Transformer is BERT~\cite{bert}, which was introduced in 2019. BERT's groundbreaking results made Transformers arguably the most popular and extensively investigated DNN architecture of the past years. The main challenge encountered when deploying transformers is the large size of these models and the resulting demand for compute resources~\cite{ivanov2020data}.
BERT includes several FFNNs of depth $2$ and weight matrices of dimensions $1024\times 4096$ and $4096\times 1024$. The major part of BERTs parameters and compute time for inference come from these FFNNs~\cite{ivanov2020data}. This makes them a prime target for pruning. We took one of these FFNNs and pruned it by removing the connections with the weights of smallest absolute value. Figure~\ref{BERT} shows the I/O-counts and the I/O lower bounds for inference on this FFNN. We considered different levels of sparsity and different cache eviction policies. Further, we assumed a fast memory of size $100$.

\begin{figure}[t]
    \centering
    \begin{subfigure}[b]{0.9\linewidth}
    \centering
    \includegraphics[width=\linewidth]{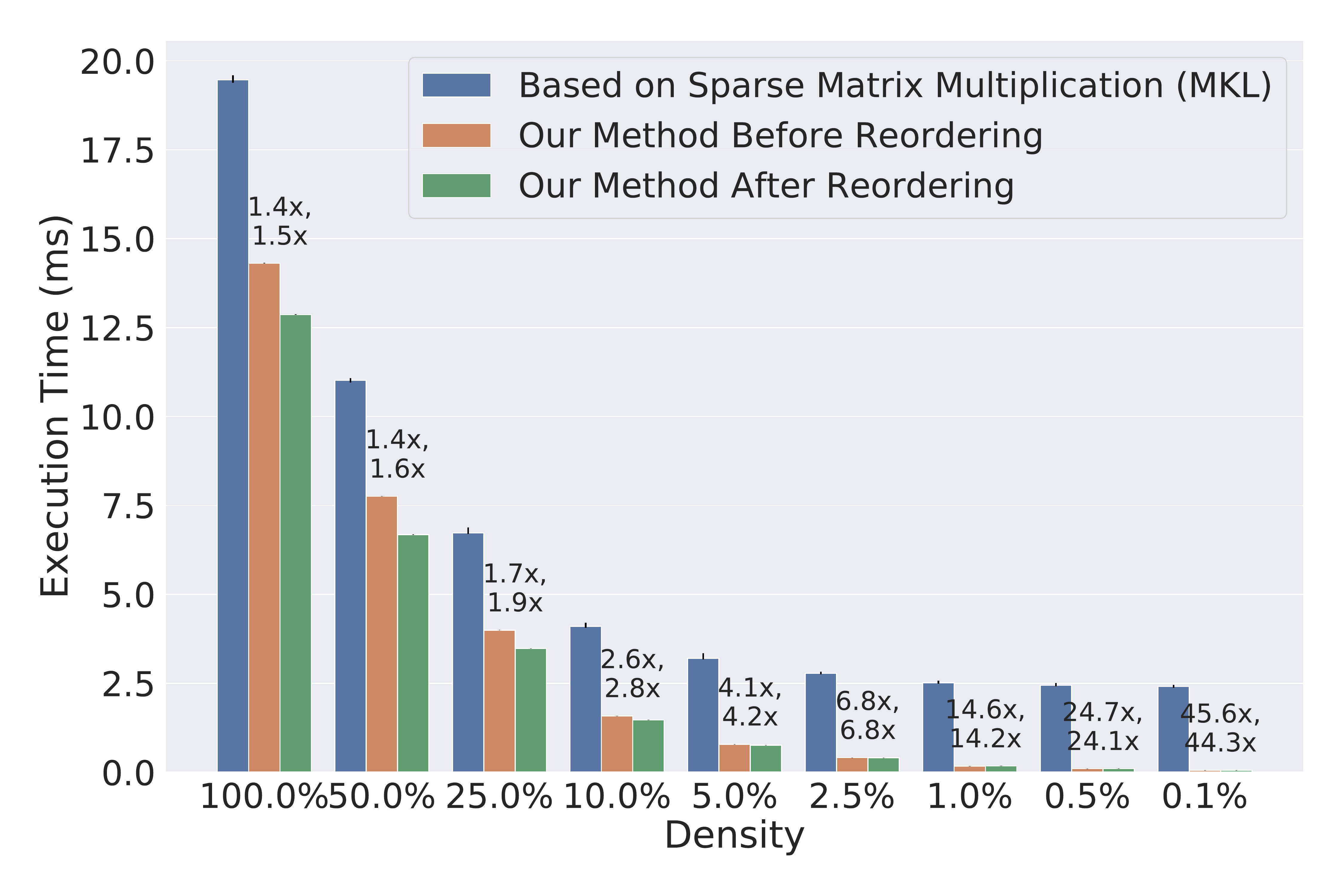}
    \vspace{-2em}
    \caption{Connection Density}
    \label{fig:realdensity}
    \end{subfigure}
    \begin{subfigure}[b]{0.9\linewidth}
    \centering
    \includegraphics[width=\linewidth]{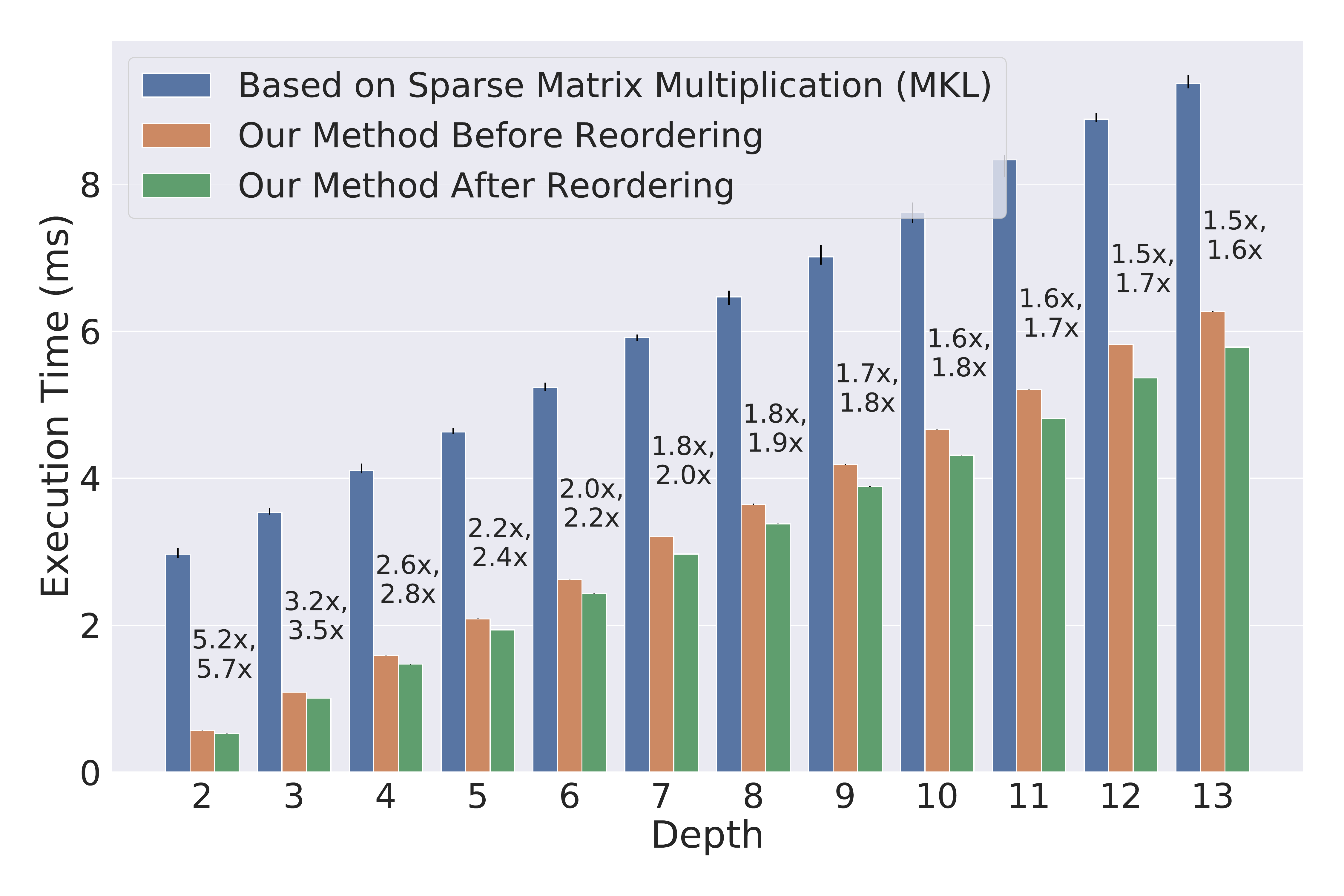}
    \vspace{-2em}
    \caption{Depth}
    \label{fig:realdepth}
    \end{subfigure}
    \begin{subfigure}[b]{0.9\linewidth}
    \includegraphics[width=\linewidth]{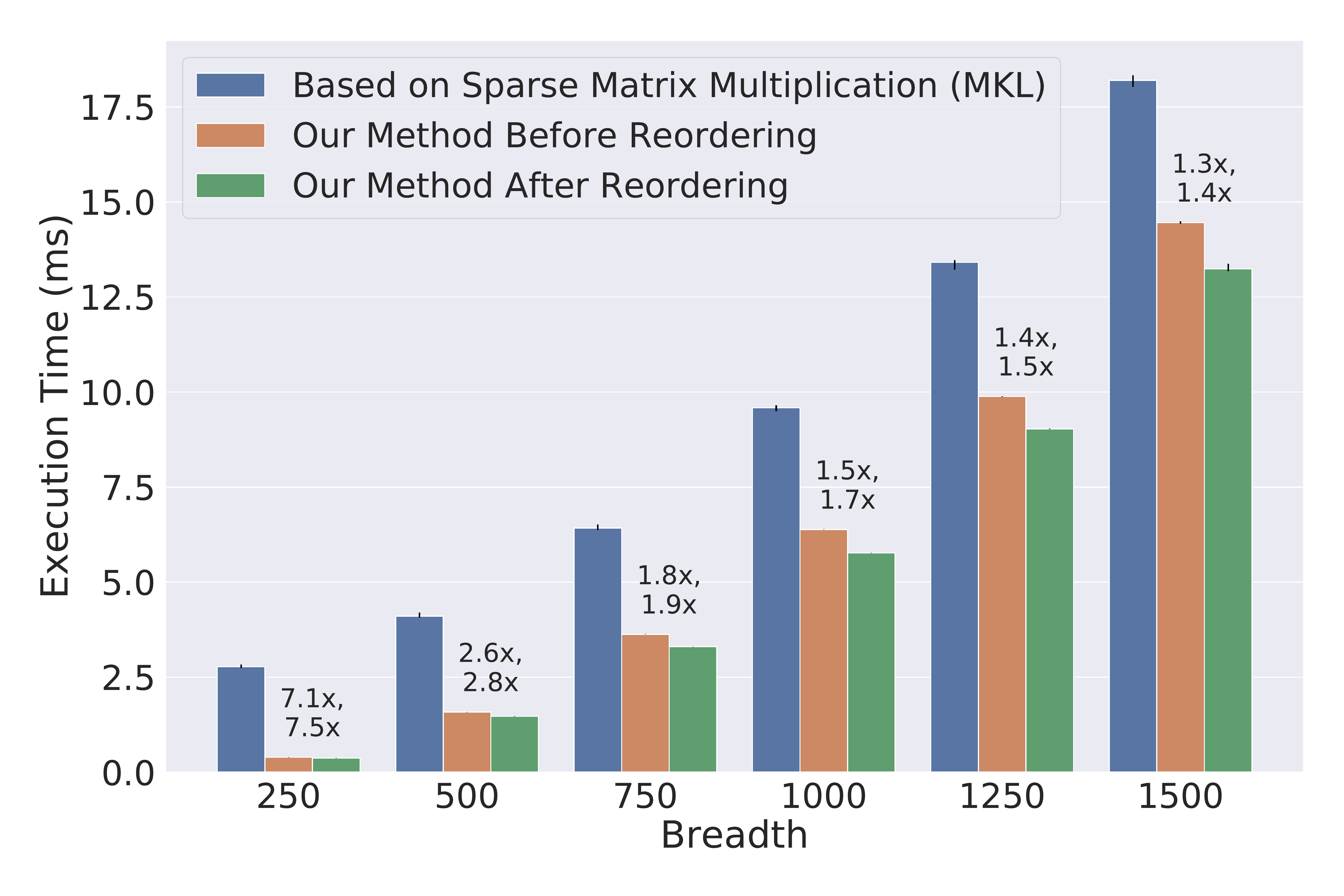}
    \vspace{-2em}
    \caption{Breadth}
    \label{fig:realbreadth}
    \end{subfigure}
    \caption{Execution time for randomly-sparse FFNNs with different methods.}
\end{figure}

\subsection{Performance Experiments}

To validate the performance benefits of I/O reduction, we perform batched inference on an Intel CPU and measure the performance. For all experiments, we reorder the FFNNs with the simulated optimal (MIN) eviction policy, and then run the FFNN before and after reordering, as well as in the traditional, layer-based approach using MKL for sparse-dense matrix matrix multiplication (CSRMM). Using batched inference (as is performed in production environments) enables the use of SIMD instructions and to better saturate the memory bandwidth. We run each experiment 10 times, with a batch size of 128, reporting median values and error bars the shortest and longest execution time. Annotations show the speed-ups that we obtain with our methods (without and with reordering) relative to the layer-based inference approach. 

The hardware used for the experiments is a 32-core Intel Xeon Gold 6130 CPU (running at 2.10 GHz) with 1.5 TB RAM. For software, we use Intel oneAPI MKL 2021.3.0 for CSRMM, and GCC 10.2.0 for our reordered networks.   

\subsubsection{Random sparse FFNNs}
Starting with the same baseline FFNN from the simulated experiments (a 10\% dense, 4-layer MLP with 500 neurons in each layer and one output neuron) we measure the execution time of our methods varying density, width, or depth while keeping the other parameters constant at their baseline value.

In Figure~\ref{fig:realdensity} we show the execution time for varying degrees of density. As we can see, the sparser the FFNN, the larger the speedup that we obtain with our method relative to the standard, layer-wise way of performing inference. For example, at density 0.1\% we get a speed-up of about 45x. On the other hand, the more sparse the FFNN, the less additional speedup we gain by reordering. In fact, somewhat surprisingly, for density values below 1\% the execution time is slightly lower without than with reordering, which we attribute to two causes: (a) the relatively low reduction in I/Os (see Section~\ref{Simux}); and (b) the small size of the networks, which makes the performance benefit negligible due to the FFNNs fitting in L2 caches. 

Another effect shown in Figure~\ref{fig:realdensity} is speedup compared with MKL on 100\% dense FFNNs. The reason for this phenomenon is that we use CSRMM in that case for consistency in measurements, rather than a dense matrix- dense matrix multiplication (GEMM), which explains the observed slowdown.

In Figures~\ref{fig:realdepth}--\ref{fig:realbreadth} we show the execution time as a function of the depth (Figure~\ref{fig:realdepth}) and breadth (Figure~\ref{fig:realbreadth}) of the FFNN (while keeping the other structural parameters constant at their baseline values). In both cases we observe speedups that grow linearly with depth and breadth (and hence, with the size of the instance), however, with diminishing returns as the network grows in one dimension and not in the others. 

In sum, our MLP performance experiments align with the simulation results, indicating that reordering is beneficial in practice for sparse neural networks, especially for cases where density is low.

\subsubsection{BERT}
In Figure~\ref{fig:realbert} we show the execution times for inference on the previously introduced encoder MLP from BERT$_\text{\textsc{LARGE}}$ with varying degrees of density.
In the case of MKL inference at 10\% density we had one clear outlier (as defined by Tukey's method~\cite{benchm}; one run took 106 ms while in the other 9 runs it took about 17 ms) which we removed from the data shown in this figure.
The figure shows that runtime after reordering is always lower than before reordering, and confirms that even for realistic network modules with only two layers, there is a benefit in using our method, especially for low densities.

\begin{figure}[t]
    \centering
    \includegraphics[width=\linewidth]{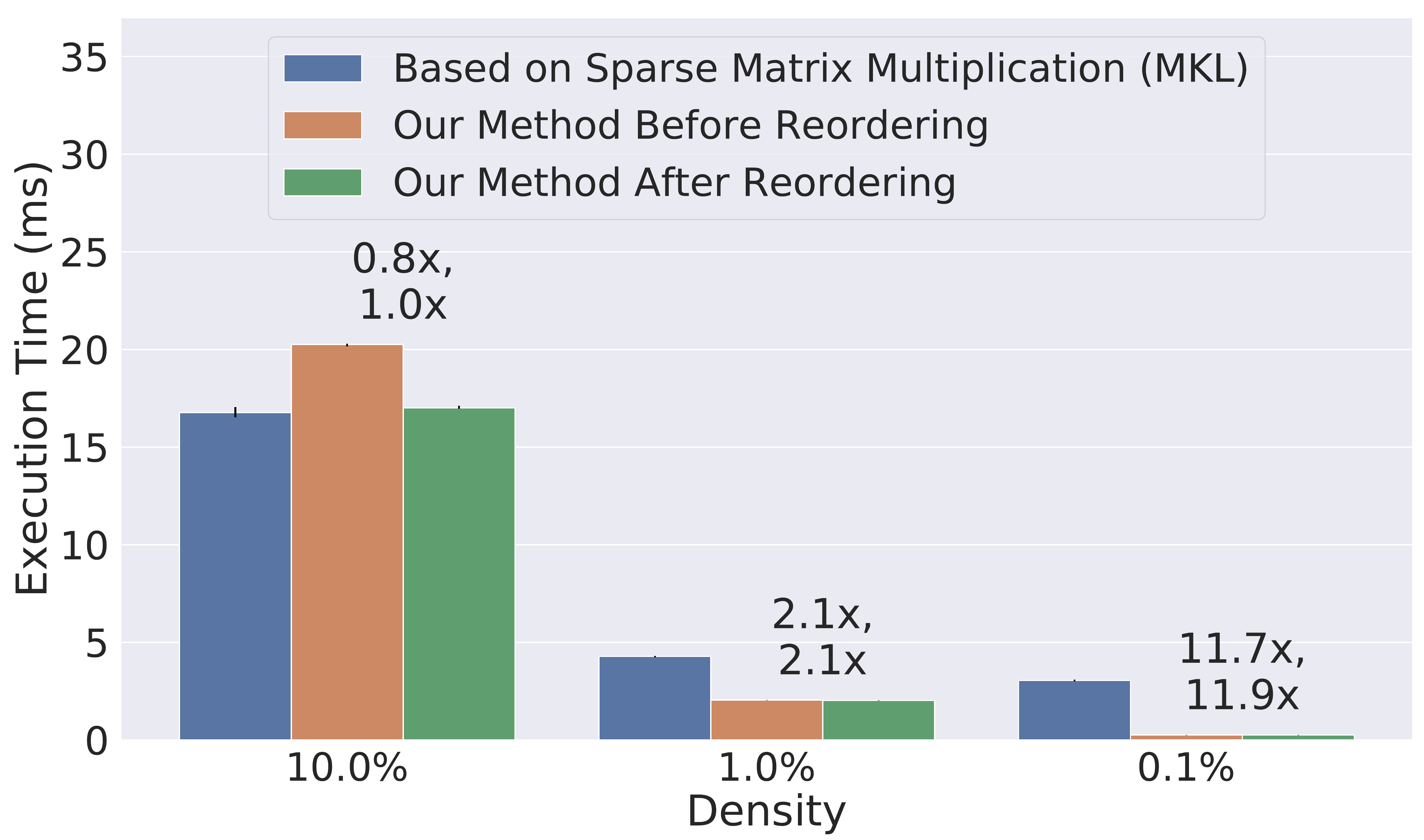}
    \vspace{-1em}
    \caption{Execution time for BERT$_\text{\textsc{LARGE}}$ with Connection Reordering and varying densities.}
    \label{fig:realbert}
\end{figure}

\section{Discussion}
\subsection{Prefetching}
Modern systems support prefetching to accelerate I/Os which is not considered in the presented model. Notice however that the effect of prefetching is ``orthogonal'' to the optimizations considered in this paper as prefetching does not reduce the overall amount of I/Os but it only affects the time the I/Os take.

\subsection{Cost of using a given topological order}
The computations of ``good'' (that is, I/O-efficient) topological orders are performed ``offline'' and once the order is determined and we perform inference, we do not need to look at it explicitly again as it is encoded in the way the connections are laid out. Hence, during inference there is no additional cost associated with processing the connections according to any given topological order.
\subsection{The term ``I/O''}
The theory in the paper applies not only to caches, but local memory of any kind (including, e.g., registers). I/O refers to a load/store from a "slow" memory, which could be represented by off-chip memory, communicating with another node, or even scratch-pad memory. The work does neither assume nor rely on caches, and covers aspects such as how to stream memory efficiently into a chip.
\section{Conclusions}

We established tight bounds on the I/O-complexity of FFNN inference. We present a 2-optimal computation strategy and show how combinatorial optimizations on the order of the connections can lead to further improvements in I/Os. 
Given a fixed memory size, Compact Growth allows us to create FFNNs on which inference can be performed I/O-optimally: we do not need to read nor write any temporary values.
Our work provides theoretical insight into the problem of mapping existing neural networks to hardware, such that energy-consumption and latency during deployment are minimized. Most importantly, these insights indicate that the improvements are achieved only when \emph{dispensing with the straightforward, layer-by-layer manner of processing DNNs}. In turn, our proposed Connection Reordering overcomes this misconception, reducing the I/O costs significantly. In our experiments on real hardware we observe speedups of up to 45$\times$ relative to the standard (CSRMM-based) way of performing inference. Further, between our 2-optimal computation strategy and the strategy that is further improved by connection reordering, we observe speedups of up to 1.17$\times$. 

\section*{Acknowledgements}
This work was supported by the EU Horizon project GLACIATION under grant agreement No. 101070141. 
T.B.N. is supported by the Swiss National Science Foundation (Ambizione Project \#185778).
We also acknowledge the Swiss National Supercomputing Centre (CSCS) for support and access to computational resources.

\bibliographystyle{unsrt}
\bibliography{references}

\begin{thebibliography}{10}

\bibitem{jia2018dissecting}
Zhe Jia, Marco Maggioni, Benjamin Staiger, and Daniele~P. Scarpazza.
\newblock Dissecting the nvidia volta gpu architecture via microbenchmarking,
  2018.

\bibitem{padal}
D.~{Unat}, A.~{Dubey}, T.~{Hoefler}, J.~{Shalf}, M.~{Abraham}, M.~{Bianco},
  B.~L. {Chamberlain}, R.~{Cledat}, H.~C. {Edwards}, H.~{Finkel},
  K.~{Fuerlinger}, F.~{Hannig}, E.~{Jeannot}, A.~{Kamil}, J.~{Keasler},
  P.~H.~J. {Kelly}, V.~{Leung}, H.~{Ltaief}, N.~{Maruyama}, C.~J. {Newburn},
  and M.~{Pericás}.
\newblock Trends in data locality abstractions for hpc systems.
\newblock {\em IEEE Transactions on Parallel and Distributed Systems},
  28(10):3007--3020, 2017.

\bibitem{szesurvey}
V.~{Sze}, Y.~{Chen}, T.~{Yang}, and J.~S. {Emer}.
\newblock Efficient processing of deep neural networks: A tutorial and survey.
\newblock {\em Proceedings of the IEEE}, 105(12):2295--2329, 2017.

\bibitem{bennunsurvey}
Tal Ben-Nun and Torsten Hoefler.
\newblock Demystifying parallel and distributed deep learning: An in-depth
  concurrency analysis.
\newblock {\em ACM Comput. Surv.}, 52(4), August 2019.

\bibitem{eie}
Song Han, Xingyu Liu, Huizi Mao, Jing Pu, Ardavan Pedram, Mark~A. Horowitz, and
  William~J. Dally.
\newblock Eie: Efficient inference engine on compressed deep neural network.
\newblock In {\em Proceedings of the 43rd International Symposium on Computer
  Architecture}, ISCA ’16, page 243–254. IEEE Press, 2016.

\bibitem{tpu}
Norman~P. Jouppi, Cliff Young, Nishant Patil, David Patterson, Gaurav Agrawal,
  Raminder Bajwa, Sarah Bates, Suresh Bhatia, Nan Boden, Al~Borchers, Rick
  Boyle, Pierre-luc Cantin, Clifford Chao, Chris Clark, Jeremy Coriell, Mike
  Daley, Matt Dau, Jeffrey Dean, Ben Gelb, Tara~Vazir Ghaemmaghami, Rajendra
  Gottipati, William Gulland, Robert Hagmann, C.~Richard Ho, Doug Hogberg, John
  Hu, Robert Hundt, Dan Hurt, Julian Ibarz, Aaron Jaffey, Alek Jaworski,
  Alexander Kaplan, Harshit Khaitan, Daniel Killebrew, Andy Koch, Naveen Kumar,
  Steve Lacy, James Laudon, James Law, Diemthu Le, Chris Leary, Zhuyuan Liu,
  Kyle Lucke, Alan Lundin, Gordon MacKean, Adriana Maggiore, Maire Mahony,
  Kieran Miller, Rahul Nagarajan, Ravi Narayanaswami, Ray Ni, Kathy Nix, Thomas
  Norrie, Mark Omernick, Narayana Penukonda, Andy Phelps, Jonathan Ross, Matt
  Ross, Amir Salek, Emad Samadiani, Chris Severn, Gregory Sizikov, Matthew
  Snelham, Jed Souter, Dan Steinberg, Andy Swing, Mercedes Tan, Gregory
  Thorson, Bo~Tian, Horia Toma, Erick Tuttle, Vijay Vasudevan, Richard Walter,
  Walter Wang, Eric Wilcox, and Doe~Hyun Yoon.
\newblock In-datacenter performance analysis of a tensor processing unit.
\newblock In {\em Proceedings of the 44th Annual International Symposium on
  Computer Architecture}, ISCA ’17, page 1–12, New York, NY, USA, 2017.
  Association for Computing Machinery.

\bibitem{redblue}
Hong Jia-Wei and H.~T. Kung.
\newblock I/o complexity: The red-blue pebble game.
\newblock In {\em Proceedings of the Thirteenth Annual ACM Symposium on Theory
  of Computing}, STOC ’81, page 326–333, New York, NY, USA, 1981.
  Association for Computing Machinery.

\bibitem{Aggar}
Alok Aggarwal and S.~Vitter, Jeffrey.
\newblock The input/output complexity of sorting and related problems.
\newblock {\em Commun. ACM}, 31(9):1116–1127, September 1988.

\bibitem{SPMDV}
Michael Bender, Gerth Brodal, Rolf Fagerberg, Riko Jacob, and Elias Vicari.
\newblock Optimal sparse matrix dense vector multiplication in the i/o-model.
\newblock volume~47, pages 61--70, 06 2007.

\bibitem{Muna}
Kameshwar Munagala and Abhiram Ranade.
\newblock I/o-complexity of graph algorithms.
\newblock In {\em Proceedings of the Tenth Annual ACM-SIAM Symposium on
  Discrete Algorithms}, SODA ’99, page 687–694, USA, 1999. Society for
  Industrial and Applied Mathematics.

\bibitem{hardapprox}
Erik~D. Demaine and Quanquan~C. Liu.
\newblock Red-blue pebble game: Complexity of computing the trade-off between
  cache size and memory transfers.
\newblock In {\em Proceedings of the 30th on Symposium on Parallelism in
  Algorithms and Architectures}, SPAA ’18, page 195–204, New York, NY, USA,
  2018. Association for Computing Machinery.

\bibitem{Liu1}
Erik~D. Demaine and Quanquan~C. Liu.
\newblock Inapproximability of the standard pebble game and hard to pebble
  graphs.
\newblock In {\em Algorithms and Data Structures}, pages 313--324, Cham, 2017.
  Springer International Publishing.

\bibitem{Liu2}
Quanquan Liu.
\newblock Red-blue and standard pebble games : Complexity and applications in
  the sequential and parallel models.
\newblock Master's thesis, Department of Electrical Engineering and Computer
  Science, MIT, Massachusetts, 2018.

\bibitem{redbluetrees}
Niels Gleinig and Torsten Hoefler.
\newblock The red-blue pebble game on trees and dags with large input.
\newblock In {\em Structural Information and Communication Complexity: 29th
  International Colloquium, SIROCCO 2022, Paderborn, Germany, June 27–29,
  2022, Proceedings}, page 135–153, Berlin, Heidelberg, 2022.
  Springer-Verlag.

\bibitem{paleo}
H.~Qi, E.~R. Sparks, and A.~Talwalkar.
\newblock Paleo: A performance model for deep neural networks.
\newblock In {\em Proc. International Conference on Learning Representations
  (ICLR)}, 2017.

\bibitem{chaos}
A.~Viebke, S.~Memeti, S.~Pllana, and A.~Abraham.
\newblock Chaos: a parallelization scheme for training convolutional neural
  networks on intel xeon phi.
\newblock {\em The Journal of Supercomputing}, 2017.

\bibitem{oyama16}
Y.~Oyama et~al.
\newblock Predicting statistics of asynchronous sgd parameters for a
  large-scale distributed deep learning system on gpu supercomputers.
\newblock In {\em IEEE International Conference on Big Data (Big Data)}, pages
  66--75, 2016.

\bibitem{DemmelConv}
James Demmel and Grace Dinh.
\newblock Communication-optimal convolutional neural nets.
\newblock {\em CoRR}, abs/1802.06905, 2018.

\bibitem{OBD}
Yann LeCun, John~S. Denker, and Sara~A. Solla.
\newblock Optimal brain damage.
\newblock In D.~S. Touretzky, editor, {\em Advances in Neural Information
  Processing Systems 2}, pages 598--605. Morgan-Kaufmann, 1990.

\bibitem{frankle}
Jonathan Frankle and Michael Carbin.
\newblock The lottery ticket hypothesis: Finding sparse, trainable neural
  networks.
\newblock 2018.

\bibitem{sparsesurvey}
Torsten Hoefler, Dan Alistarh, Tal Ben-Nun, Nikoli Dryden, and Alexandra Peste.
\newblock Sparsity in deep learning: Pruning and growth for efficient inference
  and training in neural networks.
\newblock {\em Journal of Machine Learning Research}, 22(241):1--124, 2021.

\bibitem{alexnet}
Alex Krizhevsky, Ilya Sutskever, and Geoffrey~E Hinton.
\newblock Imagenet classification with deep convolutional neural networks.
\newblock In {\em Advances in neural information processing systems}, pages
  1097--1105, 2012.

\bibitem{attention}
Ashish Vaswani, Noam Shazeer, Niki Parmar, Jakob Uszkoreit, Llion Jones,
  Aidan~N Gomez, \L~ukasz Kaiser, and Illia Polosukhin.
\newblock Attention is all you need.
\newblock In I.~Guyon, U.~V. Luxburg, S.~Bengio, H.~Wallach, R.~Fergus,
  S.~Vishwanathan, and R.~Garnett, editors, {\em Advances in Neural Information
  Processing Systems}, volume~30, pages 5998--6008. Curran Associates, Inc.,
  2017.

\bibitem{ivanov2020data}
Andrei Ivanov, Nikoli Dryden, Tal Ben-Nun, Shigang Li, and Torsten Hoefler.
\newblock Data movement is all you need: A case study on optimizing
  transformers, 2020.

\bibitem{strom2015scalable}
Nikko Strom.
\newblock Scalable distributed dnn training using commodity gpu cloud
  computing.
\newblock In {\em Sixteenth Annual Conference of the International Speech
  Communication Association}, 2015.

\bibitem{han2015deep}
Song Han, Huizi Mao, and William~J. Dally.
\newblock Deep compression: Compressing deep neural networks with pruning,
  trained quantization and huffman coding, 2015.

\bibitem{Aji_2017}
Alham~Fikri Aji and Kenneth Heafield.
\newblock Sparse communication for distributed gradient descent.
\newblock {\em Proceedings of the 2017 Conference on Empirical Methods in
  Natural Language Processing}, 2017.

\bibitem{sparcml}
Cedric Renggli, Saleh Ashkboos, Mehdi Aghagolzadeh, Dan Alistarh, and Torsten
  Hoefler.
\newblock Sparcml: High-performance sparse communication for machine learning.
\newblock In {\em Proceedings of the International Conference for High
  Performance Computing, Networking, Storage and Analysis}, SC ’19, New York,
  NY, USA, 2019. Association for Computing Machinery.

\bibitem{EDEN}
Skanda Koppula, Lois Orosa, A.~Giray Yağlıkçı, Roknoddin Azizi, Taha
  Shahroodi, Konstantinos Kanellopoulos, and Onur Mutlu.
\newblock Eden: Enabling energy-efficient, high-performance deep neural network
  inference using approximate dram.
\newblock {\em Proceedings of the 52nd Annual IEEE/ACM International Symposium
  on Microarchitecture}, Oct 2019.

\bibitem{infdist}
Kartikeya Bhardwaj, Ching-Yi Lin, Anderson Sartor, and Radu Marculescu.
\newblock Memory- and communication-aware model compression for distributed
  deep learning inference on iot.
\newblock {\em ACM Transactions on Embedded Computing Systems}, 18(5s):1–22,
  Oct 2019.

\bibitem{Sze}
Y.~{Chen}, T.~{Yang}, J.~{Emer}, and V.~{Sze}.
\newblock Eyeriss v2: A flexible accelerator for emerging deep neural networks
  on mobile devices.
\newblock {\em IEEE Journal on Emerging and Selected Topics in Circuits and
  Systems}, 9(2):292--308, 2019.

\bibitem{MINcache}
Laszlo~A. Belady.
\newblock A study of replacement algorithms for virtual-storage computer.
\newblock {\em IBM Syst. J.}, 5:78--101, 1966.

\bibitem{SA}
C~Koulamas, SR~Antony, and R~Jaen.
\newblock A survey of simulated annealing applications to operations research
  problems.
\newblock {\em Omega}, 22(1):41 -- 56, 1994.

\bibitem{bert}
Jacob Devlin, Ming{-}Wei Chang, Kenton Lee, and Kristina Toutanova.
\newblock {BERT:} pre-training of deep bidirectional transformers for language
  understanding.
\newblock {\em CoRR}, abs/1810.04805, 2018.

\bibitem{benchm}
Torsten Hoefler and Roberto Belli.
\newblock Scientific benchmarking of parallel computing systems: Twelve ways to
  tell the masses when reporting performance results.
\newblock In {\em Proceedings of the International Conference for High
  Performance Computing, Networking, Storage and Analysis}, SC '15, New York,
  NY, USA, 2015. Association for Computing Machinery.

\end{thebibliography}

\appendix

\section{Additional details on experimental setup and results}
\subsection{How we generate random FFNNs}

We describe how we generate random FFNNs of width $w$, depth $d$, and sparsity $p$:
For each non-output neuron, we determine how many outgoing connections it has, by drawing uniformly at random an integer $k$ between $1$ and $max(1,\lceil 2\cdot p\cdot\# \text{of neurons in the next layer}-1\rceil)$. Then, we connect this neuron to $k$ randomly chosen neurons of the next layer ($k\geq 1$ ensures that the FFNN is connected and the output neuron is connected to all neurons of the last hidden layer).
Once we determined all connections, we order the connections layer-by-layer with respect to their output neuron (according to the proof of Theorem \ref{oinft}, this ensures that we have a connection layout that uses not more than twice the optimal number of I/Os).

\subsection{How we generate FFNNs with Compact Growth}\label{CGD}
We start with $M_g-2$ readily computed input neurons in the bag (in our experiments we used $M_g=100,300,$ and $500$). Then we iterate from $1$ to $1000$, and in each iteration: 1) we add a new neuron 2) we draw incoming connections to this neuron from $5$ other randomly chosen neurons from the bag 3) we remove the last of these $5$ randomly chosen neurons from the bag. After the $1000$ iterations, we add one more neuron (the output neuron) and draw connections from all remaining neurons in the bag to this output neuron.

\end{document}